\documentclass[dvipsname, 11pt, letterpaper]{article}

\usepackage{algorithm}
\usepackage{algorithmic}
\usepackage{amsmath}\allowdisplaybreaks
\usepackage{amssymb}
\usepackage{amsthm}
\usepackage{color,xcolor}
\usepackage[margin = 1in]{geometry}
\usepackage{hyperref}
\usepackage{ifthen}
\usepackage{mathtools}
\usepackage{braket}
\usepackage{graphicx}
\usepackage{caption}
\usepackage{makecell,multirow}
\usepackage{subcaption}
\usepackage{framed}
\usepackage{authblk}
\hypersetup{colorlinks=true, linkcolor=red, citecolor=blue, anchorcolor=blue, urlcolor=blue}  

\DeclareMathOperator*{\argmax}{argmax}

\def \bv{\boldsymbol{v}}
\def \bhv{\hat{\boldsymbol{v}}}
\def \Mcal{\mathcal{M}}

\def \F{\mathcal{F}^{(2)}}
\def \Fl{\mathcal{F}^{(2,\ell)}}
\def \EFOt#1{EFO^{(2)}_{#1}}
\def \EFOl{EFO^{(2,\ell)}}
\def \maxV{\textsc{maxV}}
\def \X{\mathcal{X}}
\def \bb{\boldsymbol{b}}

\newtheorem{theorem}{Theorem}

\newtheorem{lemma}{Lemma}
\newtheorem{definition}{Definition}
\newtheorem{remark}{Remark}
\newtheorem{corollary}{Corollary}

\newtheorem{mechanism}{Mechanism}

\newenvironment{thmbis}[1]
  {\renewcommand{\thetheorem}{\ref{#1}}%
   \addtocounter{theorem}{-1}%
   \begin{theorem}}
  {\end{theorem}}

\newenvironment{etmechanism}[1]
  {\renewcommand{\themechanism}{$\widetilde{\ref{#1}}$}%
   \addtocounter{mechanism}{-1}%
   \begin{mechanism}}
  {\end{mechanism}}
  
\newcommand{\eq}[1]{(\ref{eq:#1})}

\newcommand{\sect}[1]{\hyperref[sect:#1]{Section~\ref*{sect:#1}}}
\newcommand{\append}[1]{\hyperref[append:#1]{Appendix~\ref*{append:#1}}}
\newcommand{\lem}[1]{\hyperref[lem:#1]{Lemma~\ref*{lem:#1}}}
\newcommand{\obs}[1]{\hyperref[obs:#1]{Observation~\ref*{obs:#1}}}
\newcommand{\thm}[1]{\hyperref[thm:#1]{Theorem~\ref*{thm:#1}}}
\newcommand{\mech}[1]{\hyperref[mech:#1]{Mechanism~\ref*{mech:#1}}}
\newcommand{\cor}[1]{\hyperref[cor:#1]{Corollary~\ref*{cor:#1}}}
\newcommand{\figg}[1]{\hyperref[fig:#1]{Figure~\ref*{fig:#1}}}
\newcommand{\tab}[1]{\hyperref[tab:#1]{Table~\ref*{tab:#1}}}
\newcommand{\ex}[1]{\hyperref[ex:#1]{Example~\ref*{ex:#1}}}
\newcommand{\defi}[1]{\hyperref[def:#1]{Definition~\ref*{def:#1}}}
\newcommand{\assump}[1]{\hyperref[assump:#1]{Assumption~\ref*{assump:#1}}}

\title{Competitive Auctions with Imperfect Predictions}
\author[1,2]{Pinyan Lu\thanks{lu.pinyan@mail.shufe.edu.cn}}
\author[3,4]{Zongqi Wan\thanks{wanzongqi20s@ict.ac.cn}}
\author[3,4]{Jialin Zhang\thanks{zhangjialin@ict.ac.cn}}

\affil[1]{ITCS, Shanghai University of Finance and Economics}
\affil[2]{Key Laboratory of Interdisciplinary Research of Computation and Economics (Shanghai University of Finance and Economics), Ministry of Education}
\affil[3]{SKLP, Institute of Computing Technology, Chinese Academy of Sciences}
\affil[4]{School of Computer Science and Technology, University of Chinese Academy of Sciences}

\date{}

\begin{document}
\maketitle
\begin{abstract}
     The competitive auction was first proposed by Goldberg, Hartline, and Wright. In their paper \cite{goldberg2001competitive}, they introduce the competitive analysis framework of online algorithm designing into the traditional revenue-maximizing auction design problem. While the competitive analysis framework only cares about the worst-case bound, a growing body of work in the online algorithm community studies the learning-augmented framework. In this framework, designers are allowed to leverage imperfect machine-learned predictions of unknown information and pursue better theoretical guarantees when the prediction is accurate(consistency). Meanwhile, designers also need to maintain a nearly-optimal worst-case ratio(robustness).

     In this work, we revisit the \emph{competitive auctions} in the learning-augmented setting. We leverage the imperfect predictions of the private value of the bidders and design the learning-augmented mechanisms for several competitive auctions with different constraints, including: \emph{digital good auctions}, \emph{limited-supply auctions}, and \emph{general downward-closed permutation environments}. We also design the learning-augmented mechanism for \emph{online auctions}. For all these auction environments, our mechanisms enjoy $1$-consistency against the strongest benchmark $OPT$, which is impossible to achieve $O(1)$-competitive without predictions. At the same time, our mechanisms also maintain the $O(1)$-robustness against all benchmarks considered in the traditional competitive analysis, including $\F, \maxV$, and $\EFOt{}$. 
     Considering the possible inaccuracy of the predictions, we provide a reduction that transforms our learning-augmented mechanisms into an error-tolerant version, which enables the learning-augmented mechanism to ensure satisfactory revenue in scenarios where the prediction error is moderate.
     We also prove a $3$ robustness ratio lower bound for mechanisms with perfect consistency. The lower bound is strictly greater than the optimal $2.42$ competitive ratio of digital good auctions, showing the impossibility of maintaining the optimal worst-case bound with perfect consistency.
\end{abstract}

\thispagestyle{empty}
\newpage
\thispagestyle{empty}
\tableofcontents
\newpage

\setcounter{page}{1}
\section{Introduction}
Revenue maximization is an important problem in the field of algorithmic mechanism design. Goldberg et al.~\cite{goldberg2001competitive} introduced the competitive analysis framework into the revenue maximization problems and they called this type of auction the \emph{competitive auctions}. In the competitive auction framework, there are $n$ bidders who want to buy an abstract service. Each bidder $i$ has a private value $v_i\in \mathbb{R}^+$ which represents his valuation of the service. The auctioneer initiates a sealed-bid auction and receives the bid $b_i$ of each bidder $i$. Then the auctioneer decides whether to serve each bidder $i$ according to a given feasibility constraint and the amount $i$ should pay for the service. The bidder aims to maximize his utility, which is the difference between $v_i$ and his payment if he gets served and $0$ if he does not get served. Our goal is to design a mechanism that maximizes the auctioneer's revenue while encouraging buyers to truthfully report their private valuations, i.e. $b_i=v_i$. That is, the mechanism is \emph{truthful} or \emph{incentive compatible}. 

Different from the Bayesian setting of the revenue maximization auctions, competitive auctions do not assume a prior distribution on private values. Instead, competitive auctions take a traditional algorithmic view of \emph{worst-case analysis}. In the framework, we compare the revenue of the mechanism with a revenue benchmark function and use the worst-case ratio (i.e. \emph{competitive ratio}) to measure the performance of the proposed mechanisms. Let $f$ denote a benchmark function, then $f$ maps $\mathcal{V}$ to $\mathbb{R}^+$, where $\mathcal{V}:=(\mathbb{R}^+)^n$ denotes the set of possible value vectors and $f$ assigns a benchmark revenue to each vector. Let $\mathcal{M}$ be a truthful mechanism and let $\mathcal{M}(\boldsymbol{v})$ denote the expected revenue of $\mathcal{M}$ when the private value vector is $\bv$. Then $\mathcal{M}$ is \emph{$\alpha$-competitive} against $f$ if
$\mathcal{M}(\bv)  \geq  \frac{f(\bv)}{\alpha}$ for all $\bv\in \mathcal{V}$.

The worst-case analysis of competitive ratio is a prominent mathematical framework for analyzing the online algorithms and it has guided the design of online algorithms for many years. However, with the development of machine learning technology in recent years, people are able to predict various unknown information. It has been found that such predictions can be applied in the design of online algorithms, resulting in better performance than worst-case bounds. There are many works that follow the line of the learning-augmented online algorithm design. Classical problems such as ski rental, caching, and scheduling are revisited in this setting \cite{purohit2018improving,anand2020customizing,lykouris2021competitive,rohatgi2020near,jiang2022online,lattanzi2020online,im2021non,li2021online}. The idea of utilizing the imperfect prediction was naturally brought to the algorithmic mechanism design field by Xu and Lu~\cite{xu2022mechanism} and Agrawal et al.~\cite{agrawal2022learning} independently. They show that the imperfect prediction of the private information of the selfish bidders can also be leveraged to improve the performance of truthful mechanisms.

The spirit of these learning-augmented algorithms is to leverage imperfect predictions robustly. That is, the algorithm should have a strong theoretical guarantee when the given prediction is accurate, but it does not lose much even when the prediction is arbitrarily bad. This is captured by the standard \emph{consistency-robustness} analysis framework proposed in \cite{purohit2018improving}. Within this framework, we assess the algorithm using two metrics: the consistency ratio and the robustness ratio. The consistency ratio is the approximation ratio when the prediction is correct and the robustness ratio is the approximation ratio when the prediction is wrong.

\subsection{Our Results}
In this paper, we investigate the learning-augmented competitive setting where the auctioneer is given the prediction $\hat{v}_i$ for the private value $v_i$ of each bidder $i$. The predicted values are public information for all bidders. Based on the different feasibility constraints of the auctioneer, we have considered three specific competitive auctions, namely \textit{digital good auctions}, \textit{limited supply auctions}, and \textit{general downward-closed permutation environments}. In the digital good auction, there is no feasibility constraint since the digital good can be duplicated infinitely and allocated arbitrarily. In limited supply auctions, the service can only be allocated to a limited number of bidders. In the general downward-closed permutation environments, the constraint is only required to be downward-closed, i.e. the subset of a feasible set is also feasible. We also consider the setting where the bidders arrive in a random order, which is called \textit{online auctions}. 

In the learning-augmented setting, our mechanism takes the prediction values as the input. In other words, our goal is to find a collection of truthful mechanisms, denoted as $\Mcal:=\{\mathcal{M}^{\bhv}\}_{\bhv\in\mathcal{V}}$, where $\bhv:=(\hat{v}_i)_{i=1}^n$ is the vector of predictions. We say $\Mcal$ is $\mu$-consistent against benchmark $f_1$, $\alpha$-robust against $f_2$ if $\mathcal{M}^{\bhv}$ is $\mu$-competitive against $f_1$ when the prediction is perfect and $\alpha$-competitive against $f_2$ when the prediction is imperfect. That is,
\begin{align}
   \Mcal^{\bhv}(\bhv) \geq \frac{f_1(\bhv)}{\mu}, \quad \Mcal^{\bhv}(\bv) \geq \frac{f_2(\bv)}{\alpha}, \quad \forall \bhv,\bv\in \mathcal{V}.
\end{align}
Notice that we are allowed to use distinct benchmarks for consistency ratio and robustness ratio, which is different from other works on learning-augmented mechanism design. 

Throughout the paper, we take the strongest benchmark $OPT(\bv):=\max_{\boldsymbol{x}\in \X}\sum_{i=1}^n x_i \cdot v_i$ as the revenue benchmark for the consistency ratio. Here $\boldsymbol{x}$ is the allocation, $\X$ is the feasibility constraint which depends on the concrete auction environment. It should be pointed out that the $OPT$ benchmark is not considered in the literature of competitive auctions without predictions since there is no truthful mechanism that can achieve a constant or even non-trivial competitive ratio against it \cite{goldberg2001competitive}. Instead, several weaker benchmarks are considered in the competitive analysis, including $\F$, $\Fl$, $\maxV$, and $\EFOt{}$. We use these weaker benchmarks as our robustness benchmarks. $\F$ is the largest revenue obtained by setting a fixed price for the service, with existing at least two bidders who can afford the price. $\maxV$ is the largest revenue obtained by a multi-unit Vickrey auction. These two benchmarks are mainly studied in digital good auctions. 
The benchmark $\F$ was extended to more complicated settings. Actually, there are two different extensions of  $\F$  in the limited-supply auctions, namely $\Fl$ and $\EFOt{}$. While $\Fl$ is more straightforward, $\EFOt{}$, the optimal revenue obtained by an envy-free allocation, is more generally applicable and economically meaningful.   $\EFOt{}$  can be further extended to the general downward-closed permutation environment.


\paragraph{Upper Bounds} If the prediction is perfect, the problem becomes purely an algorithm design problem rather than a mechanism design problem. Thus, it is not difficult to achieve perfect consistency alone. The challenging task is to achieve perfect consistency and constant robustness simultaneously. The main result of this paper is to achieve this goal for all the auction environments and benchmarks studied before. 

\begin{theorem}\label{thm:main}
  For digital good auctions, limited supply auctions, auctions with general symmetric downward-closed feasibility constraints, and online auctions, there are truthful auctions with perfect consistency with respect to the optimal offline solution and constant robustness with respect to all the benchmarks studied for competitive auctions before.  
\end{theorem}

We summarize our robustness ratio in different auction environments and for different benchmarks in \tab{results}. It is worth noting that the mechanisms we designed for digital good auction and downward-closed permutation environments are both black-box reductions. Let $\alpha$ be the competitive ratio of the employed black-box mechanism, our mechanisms achieve $(\alpha+2)$-robust for digital good auctions and $(\alpha+7)$-robust for auctions with general downward-closed permutation environments.

\begin{table}[htbp]
  \centering
    \begin{tabular}{|c|c|c|c|c|c|}
    \hline
    \textbf{Auction Environment} & \textbf{\makecell{Robustness \\Benchmark}} & \textbf{\makecell{Robustness Ratio\\(with perfect consistency)}} &\textbf{\makecell{Competitive Ratio\\(without prediction)}}\\
    \hline
    \multirow{2}{*}{Digital Good Auction}  & $\F$   & $4.42^*$~[\cor{digital good auction}] & $2.42$~\cite{chen2014optimal}\\
    \cline{2-4}
    & $\maxV$ & $e+1^*$~[\cor{digital good auction}] & $e-1$~\cite{chen2014optimal}\\
    \cline{1-4}
    \makecell{Downward-closed\\ Permutation Environment} & $\EFOt{}$ & $13.5^*$~[\cor{general}] & $6.5$~\cite{chen2015competitive}\\
    \cline{1-4}
    \multirow{2}{*}{$\ell$-Limited Supply}   & $\Fl$   & 4.42~[\thm{limited supply f}] & 2.42~\cite{chen2014optimal}\\
    \cline{2-4}
    & $\EFOt{}$ & 5.42~[\thm{limited supply efo}] & 3.42~\cite{chen2015competitive}\\
    \hline
    \multirow{2}{*}{Online Auction}  & $\F$   & $8.84$~[\thm{online auction}] & $2.42$~\cite{chen2014optimal}\\
    \cline{2-4}
    & $\maxV$ & $2e+2$~[\thm{online auction}] & $e-1$~\cite{chen2014optimal}\\
    \cline{1-4}
    \end{tabular}%
\caption{Our robustness ratio upper bound results. We omit the consistency ratio results since all results shown are with perfect consistency against the strongest benchmark $OPT$. $^*$The robustness ratio with mark means that the related mechanism is a black-box reduction.}
  \label{tab:results}%
\end{table}%

\paragraph{Lower Bound} 
In the table, we can see that there is some degradation in ratio if we compare our robustness ratios with the previous best-known competitive ratios without considering consistency. Is such degradation necessary? Can we achieve perfect consistency while still keeping the same robustness ratio? We prove that this is impossible by the following lower bound result. 

\begin{theorem}\label{thm:lower bound F}
  Any $1$-consistent truthful mechanism for digital good auction has a robustness ratio against $\F$ no less than $3$.
\end{theorem}
The competitive ratio of the optimal mechanism for digital good auction is $2.42$, thus our lower bound shows that it is impossible to maintain the optimal competitive ratio when we require the perfect consistency of a mechanism. Since the digital good auction is a special case of other auction environments, the lower bound also applies to other auction environments considered.

\paragraph{Error-Tolerant Design} The original consistency-robustness analysis framework is too extreme that the consistency guarantee only works when the prediction is perfect, which makes the theoretical results not practical. To give a practical solution, we also propose \textit{Error-Tolerant Mechanisms} for digital good auctions and the auctions with downward-closed permutation environments respectively in \sect{error-tolerant}. The error-tolerant mechanism takes a confidence level $\gamma$ as the input parameter. Let $\eta = \max_{i\in [n]}\{\frac{\hat{v}_i}{v_i},\frac{v_i}{\hat{v}_i}\}$ denote the maximal relative prediction error, if $\gamma\geq \eta$, our mechanism is $O(\lceil\log \gamma\rceil)$-competitive against the strongest $OPT$ benchmark. Otherwise, our mechanism is $O(\lceil\log \gamma\rceil)$-competitive against $\mathcal{F}^{(2)}$ or $\EFOt{}$.



\subsection{Techniques and Challenges}\label{sect:techniques}
Notice that our problem is an objective-maximization problem, such a problem admits a trivial random combined mechanism to leverage the predictions. The combined mechanism is, with probability $p$ we provide service and charge payment according to the optimal revenue calculated from $\bhv$; with probability $1-p$, we run a competitive truthful mechanism. These two mechanisms are truthful so their random combination is also truthful. With a bit of abuse of notation, we also use $OPT(\bhv)$ to represent the optimal mechanism while assuming the prediction is perfect. As an example, in the digital good auctions, $OPT(\bhv)$ mechanism is just to offer the price $\hat{v}_i$ to each bidder $i$. Assume that $\mathcal{M}$ is an $\alpha$-competitive mechanism. The trivial random combination of $OPT(\bhv)$ and $\Mcal$ is $\frac{1}{p}$-consistent and $\frac{\alpha}{1-p}$-robust. In a word, for the objective-maximization problem, it is trivial to obtain $O(1)$-consistent and $O(\alpha)$-robust mechanism if the problem admits $\alpha$-competitive mechanism without prediction. However, such a trivial random combination is far from sufficient to achieve $1$-consistency, since we need $p=1$ to obtain perfect consistency and this implies an unbounded robustness ratio. 

Our idea comes from another kind of combination of $OPT(\bhv)$ and a black-box competitive mechanism $\Mcal$. At a glance at the learning-augmented competitive auction problem, one may come up with a naive combination to achieve perfect consistency. That is, to run the $OPT(\bhv)$ mechanism when $\bb=\bhv$, and to run $\Mcal$ when $\bb\neq \bhv$. However, the combined mechanism is not truthful since the bidder can alter the mechanism selected by misreporting his value. Let $\bb_{-i}$ the bid vector without $b_i$. To ensure truthfulness, an observation is that, if we have several truthful mechanisms and we decide which mechanism to apply to bidder $i$ only based on $\bb_{-i}$ and $\bhv$, then the combined mechanism will be truthful since the bidder cannot determine the mechanism applied to it. We call this type of combined mechanism a \emph{Bid-Independent Combination}, which is the core technical component of our proof.

Back to the naive combination, we can modify it by using the bid-independent combination approach. If $\bb_{-i} = \bhv_{-i}$, indicating that the prediction is correct for all bidders except $i$, we apply the $OPT(\bhv)$ mechanism to bidder $i$. If $\bb_{-i} \neq \bhv_{-i}$, indicating that there is at least one incorrect prediction for the bidders other than $i$, we apply $\mathcal{M}$ to bidder $i$. This combined mechanism is truthful by the Bid-Independent Combination trick.  If $\bhv =\bv$, then $\bhv_{-i}=\bv_{-i}, \forall i$, every bidder is applied with $OPT(\bhv)$ mechanism, so the mechanism actually reduces to $OPT(\bhv)$ mechanism and it gains optimal revenue in this situation. Thus, the mechanism is $1$-consistent.

Through the above Bid-Independent Combination trick, we achieve truthfulness and perfect consistency freely. However, compared with the random combination and the naive untruthful combination, the Bid-Independent Combination has a very subtle issue that it may be the case that different mechanisms are operating simultaneously on the different bidders. Thus the robustness and feasibility guarantee of the combined mechanism cannot be directly inherited from $OPT(\bhv)$ and $\Mcal$. This is the main challenge we need to overcome in our paper.

In digital good auctions, we only face the robustness problem since the bid-independent combination of arbitrary mechanisms is feasible in this case. We notice that the combined mechanism runs two mechanisms simultaneously only when the number of wrong predictions is exactly one. 
It is not clear if the bid-independent combination can obtain a constant robustness ratio or not. To save it, we employ another random combination outside the bid-independent combination. 
We find that the $OPT(\bhv)$ mechanism guarantees $2$-robustness of the robustness benchmark when the number of wrong predictions is one which is exactly the case that bid-independent combination fails. Here we crucially use the fact that  $\F$ (and actually all the competitive benchmarks considered) does not depend on the highest value. Thus, one wrong prediction will not ruin the mechanism significantly.  


For more complex constraints, we must deal with the robustness and the feasibility at the same time. 
In the general downward-closed permutation environment, the feasibility constraint becomes very challenging. The difficulty we meet is that the downward-closed feasibility requirement is so general that we have no idea how to satisfy it if the mechanism  $OPT(\bhv)$ and $\Mcal$ run simultaneously. So our idea is to incorporate a \emph{rejection mechanism} $\Mcal_{\emptyset}$ into bid-independent combination to separate $OPT(\bhv)$ and $\Mcal$. Here, the rejection mechanism means the mechanism that trivially rejects all bidders. Rigorously, we apply $OPT(\bhv)$ on bidder $i$ if $\bb_{-i}=\bhv_{-i}$; apply $\Mcal_{\emptyset}$ on bidder $i$ if there is exactly one wrong prediction in $\bhv_{-i}$; apply $\Mcal$ on bidder $i$ if there is at least two wrong predictions in $\bhv_{-i}$. We can show that it is the only possible case that mechanism $\Mcal_{\emptyset}$ and $OPT(\bhv)$ run simultaneously or mechanism $\Mcal_{\emptyset}$ and $\Mcal$ run simultaneously. On the other hand, $\Mcal_{\emptyset}$ can run concurrently with any mechanism in the downward-closed constraint. By inserting the mechanism $\Mcal_{\emptyset}$, we solve the feasibility issue. However, the bid-independent combination of these three mechanisms can only guarantee a robustness ratio when the total number of wrong predictions is at least $3$. So we are still left with the case that the number of wrong predictions is $1$ and $2$.  In the case of one wrong prediction, we can rely on the intuition that our chosen benchmark is not significantly affected by a single incorrect prediction. However, this intuition no longer holds when there are two wrong predictions, as altering two values can completely change the benchmarks. To address this situation, we employ the \emph{benchmark decomposition} technique proposed in \cite{chen2015competitive}. We decompose the benchmark into the sum of a new benchmark that remains robust even with two wrong predictions and a simple benchmark $2v_2$.  While $2v_2$ is not resilient to two wrong predictions, it is simple enough so that we can design a specific mechanism to handle it.

\subsection{Related Works}
\paragraph{Competitive Auctions} The study of digital good auctions was initiated by Goldberg et al.~\cite{goldberg2001competitive}, where they propose the random sampling optimal price auction. Fiat et al.~\cite{fiat2002} and Alaei et al.~\cite{alaei2009random} showed that the auction has a competitive ratio of $15$ and $4.68$ against the benchmark $\F$ respectively. Fiat et al. proposed the random sampling cost-sharing auction and proved that it achieves $4$-competitiveness against $\F$. There was a sequence of works on improving the competitive ratio, including \cite{goldberg2003competitiveness, hartline2005optimal,feige2005competitive,ichiba2010averaging}. Goldberg et al.~\cite{goldberg2004lower} proved the $2.42$ competitive lower bound of the digital good auction. Chen et al.~\cite{chen2014optimal} proved that the optimal competitive ratio is exactly $2.42$ and they also study the $\maxV$ benchmark and show the optimal competitive ratio is $e-1$. The online random order setting of the digital good auction was considered by Koutsoupias et al.~\cite{koutsoupias2013competitive} where the authors provided a reduction from the offline setting to the online setting.

The limited supply auction was studied with respect to the benchmark $\Fl$ by Goldberg et al.~\cite{goldberg2006competitive}, where the authors showed a straightforward reduction from digital good auctions to limited supply auctions. Thus, the competitive ratio results for the digital good auctions also held for the limited supply setting. Hartline and Yan~\cite{hartline2011envy} defined the economically significant benchmark $\EFOt{}$ and showed a constant competitive ratio against it in limited supply environments and general downward-closed permutation environments. The competitive ratio in the general downward-closed permutation environment was improved by several works \cite{devanur2015envy,ha2012biased,ha2013mechanism,devanur2013prior}. So far, the best competitive ratio against $EFO^{(2)}$ is $3.42$ in the limited supply environment and $6.5$ in the downward-closed permutation environment, which was obtained in \cite{chen2015competitive} via the benchmark decomposition technique. 

\paragraph{Learning-augmented Mechanisms} Xu and Lu~\cite{xu2022mechanism} revisited several mechanism design problems in the learning-augmented settings, including single-item revenue maximization, frugal path auction, truthful job scheduling, and two facility location on a line. The revenue maximization auction they studied is different from ours, not only because the constraint is restricted to single-item, but also because they compete with the $OPT$ benchmark thus the robustness ratio is not constant and depends on the scaling of the bids. Balkanski et al.~\cite{balkanski2023strategyproof} studied truthful job scheduling independently and attained the tight asymptotic consistency-robustness trade-off. Agrawal et al.~\cite{agrawal2022learning} designed the learning-augmented mechanism for the one-facility location problem in the two-dimensional Euclidean space. Istrate and Bonchis~\cite{istrate2022mechanism} studied the learning-augmented mechanism for the obnoxious facility location problem on several metric spaces. Gkatzelis et al.~\cite{gkatzelis2022improved} showed the ability of imperfect predictions to improve the Price of Anarchy in games. Balkanski et al.~\cite{balkanski2023online} studied the learning-augmented online mechanism. A very recent and relevant work by  Balcan et al.~\cite{balcan2023bicriteria} investigated the general multi-dimensional mechanism design problem with side information (predictions). The authors proposed a mechanism that achieves an $O(\log(H))$-consistency ratio, where $H$ represents an upper bound on any bidder's value for any allocation. To make a comparison between our paper and the one of Balcan et al.~\cite{balcan2023bicriteria}, they considered more general auction environments and the maximization problem of both social welfare and revenue simultaneously, our paper only focuses on the single-parameter auction environment and revenue maximization. When applying their mechanism to our problem, it achieves $O(\log(H))$-consistency and $O(\log(H))$-robustness, and our mechanism achieves $1$-consistency and $O(1)$-robustness. Moreover, their robustness benchmark is the revenue obtained by the VCG mechanism, which is a far weaker benchmark in the problems considered in our paper. For example, the revenue of the VCG mechanism applying on any digital good auction is $0$, making this benchmark meaningless here.



\section{Preliminary}\label{sect:preliminary}
\subsection{Auction Environments}\label{sect:auction environment}
\paragraph{Notations and the Indices} We use $\boldsymbol{v}:=(v_1,v_2,\ldots,v_n)$ and $\hat{\boldsymbol{v}}:=(\hat{v}_1,\hat{v}_2,\ldots,\hat{v}_n)$ to denote the private value vector and predicted value vector respectively. We arrange the indices such that $v_1\geq v_2\geq \cdots \geq v_n$, and we define $\sigma:[n]\rightarrow [n]$ to be the order of predicted values, such that $\hat{v}_{\sigma(i)}\geq \hat{v}_{\sigma(i+1)}$ for all $i$. When there are bidders with the same private value or predicted value, we break ties to ensure consistency between the two orders. That is, if $v_i=v_j, \hat{v}_i=\hat{v}_j$, and $i<j$, then we set $\sigma(i)<\sigma(j)$, and vice versa.

In addition, we allow for random mechanisms since deterministic mechanisms have been shown to not have a non-trivial competitive ratio against some meaningful benchmarks. A random mechanism is said to be truthful if it is given by a distribution of truthful deterministic mechanisms.

In the main text, we consider only the following two auction environments: \textit{digital good auctions} and \textit{general downward-closed permutation environments}. 

\paragraph{Digital good auction} The auctioneer can offer an unlimited number of services and any combination of bidders can be served simultaneously~(e.g., digital goods).

\paragraph{$\ell$-limited supply auction} At most $\ell$ bidders can be served simultaneously. For example, the auctioneer has $\ell$ identical items to sell, and each bidder needs at most one item. 

\paragraph{General downward-closed permutation environments} The feasibility constraint of the allocation is a general \emph{symmetric} \textit{downward-closed} set $\X\subseteq [0,1]^n$. A set $\X$ is said to be downward-closed if for every $\boldsymbol{y}\leq \boldsymbol{x}$ and $\boldsymbol{x}$ feasible, we have that $\boldsymbol{y}$ is also feasible. We say that $\X$ is symmetric if, for any $\boldsymbol{x}\in \X$, the entries of $\boldsymbol{x}$ can be arbitrarily permuted without affecting its feasibility. We also assume that $\X$ is convex since we consider randomized mechanisms and the random combination of two feasible allocations is also feasible. We focus on the symmetric constraint mainly because we can only define the envy-free optimal revenue benchmark on a symmetric constraint. Notice this auction environment includes the digital good auction and $\ell$-limited supply auction, where $\mathcal{X}=\{\boldsymbol{x}\mid \sum_{i=1}^n x_i\leq \ell, x_i\in [0,1], \forall i\in [n]\}$, as its special case. But the environment is strictly more general than $\ell$-limited supply auctions. In fact, given any downward-closed subset system $\mathcal{I}\subseteq 2^{[n]}$ (which may not be symmetric), we can define its symmetric version as follows. Let $\mathbb{S}_n$ be the symmetric group of $[n]$, $\mathtt{Conv}(\mathcal{I})\subseteq [0,1]^n$ be the convex hull of $\mathcal{I}$. For $\delta\in \mathbb{S}_n$, let $\mathcal{I}_{\delta}$ be the subset system obtained by permuting the elements in $\mathcal{I}$ with $\delta$. Then $\mathcal{X}:= \sum_{\delta\in\mathbb{S}_n}\frac{1}{n!}\cdot\mathtt{Conv}(\mathcal{I}_{\delta})$ is a symmetric downward-closed set. Here the sum between sets is the Minkowski Sum. This is usually called a \emph{permutation environment} \cite{hartline2011envy}.

\paragraph{Online auctions} Bidders arrive online in random order. Upon the arrival of each bidder $i$, the prediction $\hat{v}_i$ is revealed to the auctioneer, then the auctioneer posts an irrevocable price for the services based on the bids and predictions of previously arrived bidders. There is no feasibility constraint, which is the same as the digital good auctions.

\subsection{Benchmarks}


\paragraph{Consistency Benchmark} Throughout the paper, we employ the optimal revenue $OPT_{\X}(\bv) = \max_{\boldsymbol{x}\in \X}\sum_{i=1}^n x_i\cdot v_i$
as the benchmark for consistency ratio, where $\X$ represents the feasibility constraint and varies across different auction environments. Specifically, in digital good auctions, $OPT_{\X}(\bv) = \sum_{i\in [n]}v_i$. With clear context, we may omit the subscript $\X$ and just write $OPT(\bv)$.

\paragraph{Robustness Benchmark} Regarding the benchmark of robustness ratio, we investigate several candidates who are widely considered in the competitive analysis. The first is 
\begin{align}
    \F(\bv) = \max_{k\geq 2} k \cdot v_k,
\end{align}
where $v_k$ is reordered as we have mentioned before. $\F$ is the largest revenue obtained by setting a fixed price for the service, with at least two bidders can afford the price. $\F$ benchmark is only meaningful in digital good auctions and online auctions which have no feasibility constraint.
One may feel confused about the additional assumption that at least two bidders get served. This is because any truthful mechanism has an arbitrarily bad competitive ratio if $v_1\gg v_2$, which makes the benchmark uninteresting without this assumption. Another common-used benchmark for digital good auctions and online auctions is $\maxV(\bv) = \max_{k<n} k\cdot v_{k+1}$, which is the largest revenue obtained by a multi-unit Vickrey auction.

For the more general downward-closed permutation environment, Hartline and Yan~\cite{hartline2011envy} introduced a new benchmark named $EFO$, which denotes the maximum revenue that can be obtained through an envy-free allocation. Similarly to the consideration of $\F$, people actually consider the competitive ratio with respect to $\EFOt{}(\bv) = EFO(\bv^{(2)})$ where $\bv^{(2)}:= (v_2,v_2,v_3,\ldots,v_n)$ is the vector obtained by replacing $v_1$ with $v_2$. $\EFOt{}$ is defined on symmetric constraint set $\X\subseteq [0,1]^n$, we use the notation $\EFOt{\X}$ to specify the constraint. Given a monotone allocation $\boldsymbol{x}:=(x_1,x_2,\ldots,x_n)$ where the monotone allocation means $x_i\geq x_{i+1},\forall i$, its envy-free revenue is defined as
$
    EF^{\boldsymbol{x}}(\boldsymbol{v}) = \sum_{i\in [n]}\sum_{j\geq i}^n v_j(x_j-x_{j+1}).
$
Then $EFO_{\X}(\bv)$ is the largest possible envy-free revenue while $\boldsymbol{x}\in \X$. Hartline and Yan~\cite{hartline2011envy} also give a characterization of envy-free payment through \textit{virtual value} so that we can better understand the $EFO$ benchmark.
\begin{definition}[Virtual value]\label{lem:virtual revenue}
    Given a vector $\bv$, we assume that the entries of $\bv$ have been sorted in descending order. We can calculate the \textit{virtual value} $\phi_i$ of each bidder corresponding to $\boldsymbol{v}$ as follows. First, we calculate the smallest non-decreasing concave function $R$ satisfying $R(j) \geq j\cdot v_j, \forall j \in [n]$ and $R(0)=0$. The formal definition is
    \begin{align}
       R(j) = \max_{1\leq l\leq j}\max_{\genfrac{}{}{0pt}{}{i,k:}{ 1\leq i\leq l\leq k\leq n}} \left[i\cdot v_i \frac{k-l}{k-i} + k\cdot v_k\frac{l-i}{k-i}\right]. 
    \end{align}
    Then the virtual value is,
    \begin{align}
        \phi_i = R(i)-R(i-1).
    \end{align}
\end{definition}

\begin{figure}
    \centering
    \includegraphics[width=0.5\textwidth]{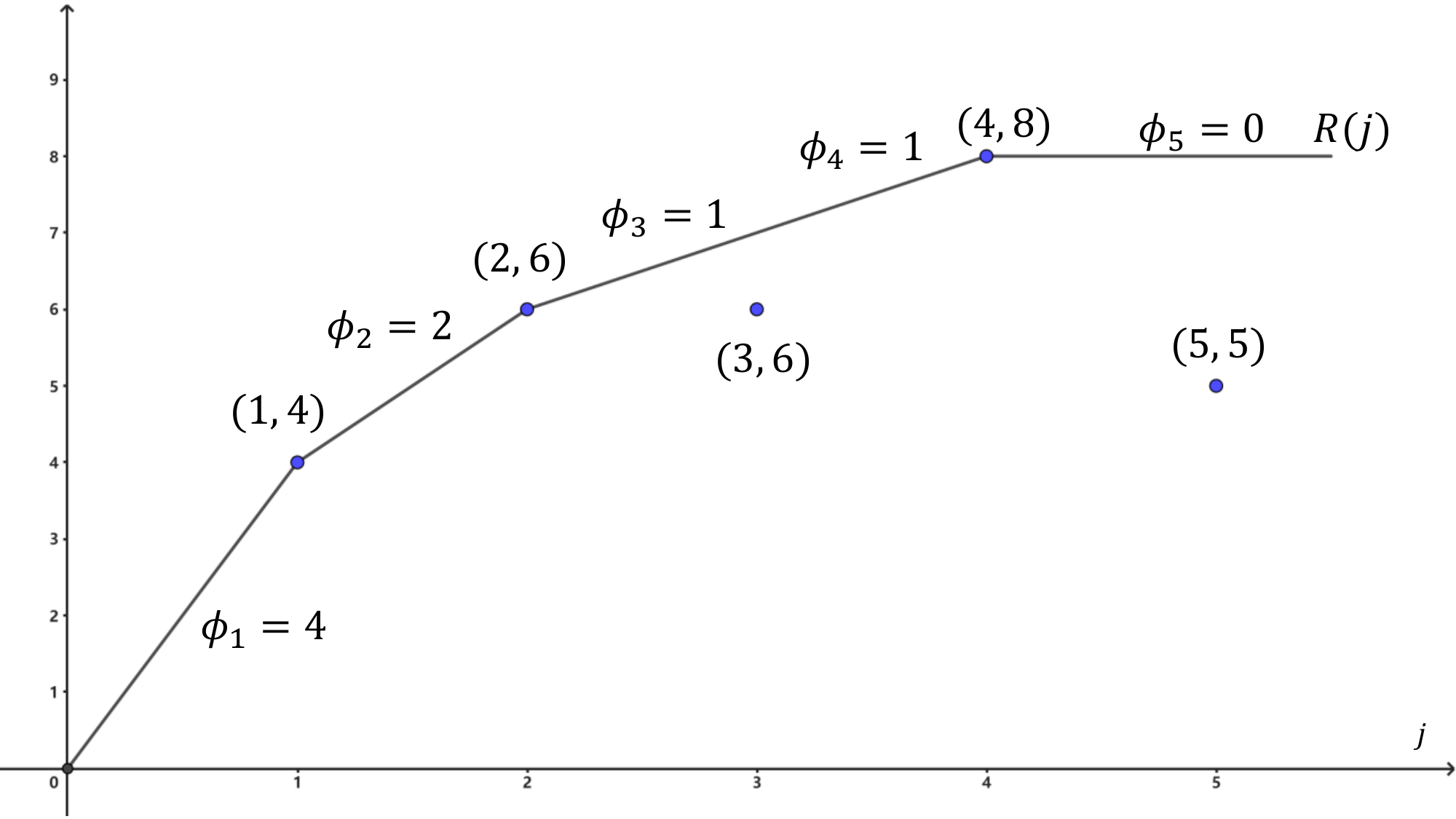}
    \caption{An example of virtual values with bidder vector $\bv=(4,3,2,2,1)$. The points on the graph are $(j,j v_j)$ and the line represents $R(j)$.}
    \label{fig:virtual value}
\end{figure}

For ease of understanding, we show an example of the virtual values with $\bv=(4,3,2,2,1)$ in Fig. \ref{fig:virtual value}.

\begin{lemma}[Restatement of Theorem 2.4 in \cite{hartline2011envy}]\label{lem:characterization of envy-free payment}
    Let $\X \subseteq [0,1]^n$ be a symmetric set and $\phi$ is the virtual value of $\bv$, then the envy-free optimal revenue over $\mathcal{X}$ is
    \begin{align}
        EFO_{\mathcal{X}}(\boldsymbol{v}) = \max_{\boldsymbol{x}\in\mathcal{X}}\sum_{i = 1}^n \phi_i\cdot x_i .
    \end{align}
\end{lemma}
Hence, $EFO_{\X}(\boldsymbol{v})$ can be seen as the optimal solution to a linear maximization problem over $\X$ with virtual values as linear weights. So far we always assume that $\bv$ is in descending order, but this is only for simplicity. If a value vector is not descending ordered, its envy-free optimal value is calculated by first ordering it and then using \lem{characterization of envy-free payment}. In other words, $EFO_{\X}(\bv)$ is a symmetric function and the identification of the bidder doesn't influence the benchmark.
\subsection{Benchmark Decomposition Lemma}
Another tool apart from the bid-independent combination trick is the following benchmark decomposition lemma. With this lemma, we decompose the benchmark into parts that correspond to different scenarios, then we can handle them separately.
\begin{lemma}[Benchmark Decomposition with Prediction]\label{lem:benchmark decomposition}
    Let $\mathcal{M}_1$ and $\Mcal_{2}$ be truthful mechanism and both with $1$-consistency ratio against $OPT$ benchmark. If $\Mcal_1$ is $\alpha$-robust against $f_1$ benchmark and $\Mcal_2$ is $\beta$-robust against $f_2$ benchmark.
    Then there is a truthful mechanism, which is $1$-consistent against $OPT$ and $(\alpha+\beta)$-robust against $f_1+f_2$ benchmark.
\end{lemma}
\begin{proof}
    Consider the mechanism which runs $\mathcal{M}_1$ with probability $\frac{\alpha}{\alpha+\beta}$ and runs $\mathcal{M}_2$ with probability $\frac{\beta}{\alpha+\beta}$. Let $\mathcal{M}_1(\boldsymbol{v})$ and $\mathcal{M}_2(\boldsymbol{v})$ denote the total revenue generated by running $\mathcal{M}_1$ and $\mathcal{M}_2$ on instance $\boldsymbol{v}$, respectively.
    
    If the prediction is perfect, then $\mathcal{M}_1(\boldsymbol{v})=\mathcal{M}_2(\boldsymbol{v})=OPT(\boldsymbol{v})$ by the consistency ratio of them. So the revenue of the combined mechanism is $OPT(\boldsymbol{v})$. Thus the combined mechanism is $1$-consistent.
    
    If the prediction is not perfect, then $\mathcal{M}_1(\boldsymbol{v})\geq \frac{f_1(\boldsymbol{v})}{\alpha}$, $\mathcal{M}_2(\boldsymbol{v})\geq \frac{f_2(\boldsymbol{v})}{\beta}$. So the total revenue of the combined mechanism is $\frac{\alpha}{\alpha+\beta}\frac{f_1(\boldsymbol{v})}{\alpha} + \frac{\beta}{\alpha+\beta}\frac{f_2(\boldsymbol{v})}{\beta}=\frac{f_1(\boldsymbol{v})+f_2(\boldsymbol{v})}{\alpha+\beta}$ which means that the mechanism is $(\alpha+\beta)$-robust against $f_1+f_2$ benchmark.
\end{proof}
Chen et al.~\cite{chen2015competitive} first proposes the benchmark decomposition lemma in the competitive analysis framework, \lem{benchmark decomposition} can be seen as a corresponding lemma in the consistency-robustness analysis framework.

\section{Digital Good Auction}\label{sect:digital good auction}
In this section, we will discuss digital good auctions against $\F$ and $\maxV$ benchmarks. Both $\F$ and $\maxV$ enjoy a property that is key to the analysis of the robustness ratio, they are both \textit{dominated} by $OPT^{(2)}(\bv) = OPT((v_2,v_2,v_3,\ldots,v_n))$.
\begin{definition}[Benchmark domination]
 For two benchmark $f_1$ and $f_2$ over $\mathcal{V} := (\mathbb{R}^+)^n$, we say $f_1$ dominates $f_2$ if $f_1(\boldsymbol{v})\geq f_2(\boldsymbol{v}), \forall \boldsymbol{v}\in \mathcal{V}$.
\end{definition}
Our result applies to all benchmarks dominated by $OPT^{(2)}$, specifically, we have the following theorem.

\begin{theorem}\label{thm:digital good auction}
    If benchmark $f$ is dominated by $OPT^{(2)}$ and there is an $\alpha$-competitive truthful mechanism against $f$, then there is a truthful mechanism, which is $1$-consistent against $OPT$ and $(\alpha+2)$-robust against $f$.
\end{theorem}
Let $\#(\boldsymbol{v}\neq \hat{\boldsymbol{v}})$ denote the number of wrong predictions. To prove the theorem, we decompose 
\begin{align}
    f(\boldsymbol{v}) = f(\boldsymbol{v})\cdot \mathbb{I}[\#(\boldsymbol{v}\neq \hat{\boldsymbol{v}})\leq 1]+f(\boldsymbol{v})\cdot \mathbb{I}[\#(\boldsymbol{v}\neq \hat{\boldsymbol{v}})\geq 2].
\end{align}
We design $1$-consistent mechanisms for the above two benchmarks respectively. Then we use \lem{benchmark decomposition} to obtain a mechanism for $f(\boldsymbol{v})$ benchmark. The mechanism for $f(\bv)\cdot \mathbb{I}[\#(\boldsymbol{v}\neq \hat{\boldsymbol{v}})\geq 2]$ is obtained via the bid-independent combination trick.

\begin{mechanism}\label{mech:dga 2}
    We first run the black-box mechanism $\mathcal{M}$ and calculate the allocation rule and payment rule of every bidder. For each bidder $i$, if $\boldsymbol{v}_{-i} = \hat{\boldsymbol{v}}_{-i}$, then offer price $\hat{v}_i$ to bidder $i$; if $\boldsymbol{v}_{-i} \neq \hat{\boldsymbol{v}}_{-i}$, we apply the black-box mechanism $\mathcal{M}$ on bidder $i$.
\end{mechanism}
\mech{dga 2} is a bid-independent combination of $OPT(\bhv)$ and $\mathcal{M}$, thus it is $1$-consistent and truthful. Moreover, it maintains the robustness of $\mathcal{M}$ when against the benchmark $f(\bv)\cdot \mathbb{I}[\#(\boldsymbol{v}\neq \hat{\boldsymbol{v}})\geq 2]$.
\begin{lemma}\label{lem:dga 2}
    If $\Mcal$ is $\alpha$-competitive against $f$, then \mech{dga 2} is truthful, and it is $1$-consistent against $OPT(\bv)$ and $\alpha$-robust against $f(\bv)\cdot \mathbb{I}[\#(\boldsymbol{v}\neq \hat{\boldsymbol{v}})\geq 2]$.
\end{lemma}
\begin{proof}
    We only prove the robustness ratio. When $\#(\boldsymbol{v}\neq \hat{\boldsymbol{v}})\geq 2$, for any bidder $i$, there is at least one incorrect prediction other than $i$. Therefore \mech{dga 2} runs $\Mcal$ for every bidder, thus its revenue is at least $\frac{1}{\alpha}f(\bv)$ by the competitive ratio of $\Mcal$. Then \mech{dga 2} is $\alpha$-robust against $f(\bv) \cdot \mathbb{I}[\#(\boldsymbol{v}\neq \hat{\boldsymbol{v}})\geq 2]$.
\end{proof}

We then use the $OPT(\bhv)$ mechanism to deal with the robustness in the case of one wrong prediction value. 
\begin{mechanism}\label{mech:dga 1}
    Use $OPT(\bhv)$ mechanism, that is, offer price $\hat{v}_i$ to bidder $i$.
\end{mechanism}

\begin{lemma}\label{lem:dga 1}
    \mech{dga 1} is truthful and it is $1$-consistent against $OPT(\bv)$ and $2$-robust against $f(\bv)\cdot \mathbb{I}[\#(\boldsymbol{v}\neq \hat{\boldsymbol{v}})\leq 1]$.
\end{lemma}
\begin{proof}
    The truthfulness and consistency are obvious. We only prove the robustness ratio. If $\#(\boldsymbol{v}\neq \hat{\boldsymbol{v}})=1$, let $j$ be the bidder with the wrong prediction. Since the predictions for bidders other than $j$ are correct, they will accept the price. Thus the revenue is $\sum_{i\neq j} v_i\geq \sum_{i\geq 2} v_i=OPT^{(2)}(\boldsymbol{v})-v_2$. Since $v_2<\sum_{i\geq 2} v_i$, we have $OPT^{(2)}(\bv)=v_2+\sum_{i\geq 2}v_i\leq 2\sum_{i\geq 2}v_i$. Therefore, we have
        \begin{align*}
            \sum_{i\neq j} v_i\geq \sum_{i\geq 2} v_i\geq \frac{1}{2}OPT^{(2)}(\boldsymbol{v})\geq \frac{1}{2}f(\boldsymbol{v}).
        \end{align*}
    The last inequality is because $OPT^{(2)}$ dominates $f$. Thus \mech{dga 1} is $2$-robust against $ f(\boldsymbol{v})\cdot \mathbb{I}[\#(\boldsymbol{v}\neq \hat{\boldsymbol{v}})\leq 1]$.
\end{proof}

Then \thm{digital good auction} is a direct corollary of \lem{dga 2}, \lem{dga 1}, and \lem{benchmark decomposition}. In \cite{chen2014optimal}, the authors show the existence of a $2.42$-compeitive mechanism for $\F$ benchmark and an $(e-1)$-competitive mechanism for $\maxV$ benchmark, which leads to the following corollary.

\begin{corollary}\label{cor:digital good auction}
    There is a truthful mechanism, which is $1$-consistent against $OPT$ and $4.42$-robust against $\mathcal{F}^{(2)}$-benchmark. Moreover, there is a truthful mechanism, which is $1$-consistent against $OPT$ and $(e+1)$-robust against $\maxV$ benchmark.
\end{corollary}

\section{\texorpdfstring{$\ell$}{}-Limited Supply Auction}\label{sect:lsa}

In this section, we will discuss $\ell$-limited supply auction where at most $\ell$ bidders can be served simultaneously. For example, the auctioneer has $\ell$ identical items to sell, and each bidder needs at most one item. 

In this setting, our benchmark for consistency is $OPT_{\X}(\bv) = \sum_{i=1}^{\ell}v_i$. For robustness, we consider the $\Fl$ and $\EFOl$ benchmarks as our benchmarks. Notice that we have defined $\EFOl$ in Section \ref{sect:preliminary}. $\Fl$ is a natural extension of $\mathcal{F}^{(2)}$, defined as 
\begin{align}
    \Fl(\bv) = \max_{2\leq k\leq l} k\cdot v_k.
\end{align}

For $\Fl$ benchmark, Goldberg et al.~\cite{goldberg2006competitive} show the equivalence between digital good auctions and limited-supply auctions. The key observation is that $\Fl$ only depends on the highest $\ell$ bidders, so we can throw away the bids except for the highest $\ell$ and run a digital good auction on these $\ell$ bidders charging the threshold bidding. This reduction maintains the competitive ratio of digital good auctions. The equivalence also exists in our setting, so the following theorem is obvious.

\begin{theorem}\label{thm:limited supply f}
    If there exists a digital good auction with $1$-consistency ratio against $OPT$, $\alpha$-robustness ratio against $\F$, then there exists a $\ell$-limited supply auction with the same consistency and robustness against $\Fl$.
\end{theorem}

When the constraint $\X$ corresponds to the $\ell$-limited supply environment, we use $\EFOl$ to substitute $\EFOt{\X}$. The $\EFOl$ benchmark depends on all bids, so we cannot use the aforementioned reduction. We use the bid-independent combination techniques again to design a $1$-consistent mechanism for this problem. Different from the unlimited supply setting, we must be careful with the feasibility here. The benchmark decomposition is the same as in the digital good auctions.
\begin{align}\label{eq:lsa decomposition}
    \EFOl = \EFOl\cdot \mathbb{I}[\#(\bv\neq\bhv)\leq 1] + \EFOl\cdot \mathbb{I}[\#(\bv\neq\bhv)\geq 2].
\end{align}

Our mechanism for $\EFOl\cdot \mathbb{I}[\#(\bv\neq\bhv)\geq 2]$ also employs a black-box competitive mechanism of $\ell$-limited supply auction.
\begin{mechanism}\label{mech:lsa 2}
 For each bidder $i$, if $\boldsymbol{v}_{-i} = \hat{\boldsymbol{v}}_{-i}$, apply $OPT_{\X}(\bhv)$ mechanism on it; if $\boldsymbol{v}_{-i} \neq \hat{\boldsymbol{v}}_{-i}$, apply the black-box mechanism $\mathcal{M}$ on bidder $i$.
\end{mechanism}

Note that different bidders may apply different mechanisms in \mech{lsa 2}, so it is not naturally a feasible mechanism. Fortunately,  we can show the feasibility holds if the employed black-box $\Mcal$ satisfies an extra assumption.
\begin{lemma}\label{lem:lsa 2}
    If the winner set of $\mathcal{M}$ is always a subset of the highest $\ell$ bidders and it is $\alpha$-competitive against $\EFOl$, then \mech{lsa 2} is feasible, $1$-consistent against $OPT$ and $\alpha$-robust against $\EFOl \cdot \mathbb{I}[\#(\bv\neq\bhv)\geq 2]$. 
\end{lemma}
\begin{proof}
    \textbf{Feasibility:} If there are more than $1$ wrong predictions or there is no wrong prediction, \mech{lsa 2} reduces to $\Mcal$ and $OPT(\bhv)$ mechanism, respectively. Thus the only situation we need to investigate is when there is exactly one wrong prediction. Denote the bidder with the wrong prediction as $j$. If $\hat{v}_j>v_j$ or the predicted value of bidder $j$ is not the highest $\ell$ of all the predictions (i.e. $\sigma^{-1}(j)>\ell$), we know the bidder is not served in $OPT(\bhv)$ mechanism since he will reject the price and does not appear in the winner set. In this case, the total number of services allocated by \mech{lsa 2} is at most $\ell$ since other bidders use the same allocation rule output by $\mathcal{M}$, which is feasible. If $\hat{v}_j< v_j$ and $\sigma^{-1}(j)\leq \ell$, then bidder $j$ is allocated with one item by applying $OPT(\bhv)$ mechanism. But in this case, since the predicted value of $j$ is one of the highest $\ell$ predicted value, $v_j>\hat{v}_j$ and for $i\neq j$, $v_i=\hat{v}_i$, we can derive that the bidder $j$ also has the highest $\ell$ true value. By the assumption that the winner set of $\mathcal{M}$ is a subset of the highest $\ell$ bidders, we know the union of $j$ and the winner set of $\mathcal{M}$ is also a subset of the highest $\ell$ bidders. So in this case, we also have the total number of services allocated by \mech{lsa 2} is at most $\ell$. This finalizes our proof of feasibility.

      \textbf{Consistency:} The consistency ratio is obviously $1$.
    
      \textbf{Robustness:} If there are at least $2$ wrong predictions, \mech{lsa 2} degenerates to $\Mcal$. The expected revenue is at least $\frac{1}{\alpha} \EFOl(\bv)$ in this case by the competitive ratio of $\Mcal$. Then \mech{lsa 2} is $\alpha$-robust against $\EFOl \cdot \mathbb{I}[\#(\bv\neq\bhv)\geq 2]$.
\end{proof}

The mechanism for $\ell$-limited supply auctions proposed in \cite{chen2015competitive} meets our assumption about the black-box mechanism in \lem{lsa 2}.
\begin{lemma}\label{lem:employee}
    The winner set of the $3.42$-competitive mechanism for $\ell$-supply auction in \cite{chen2015competitive} is always a subset of the highest $\ell$ bidders.
\end{lemma}
\begin{proof}
    The mechanism is the random combination of the $\ell$-items Vickrey auction and the optimal mechanism for $\ell$-bidders digital good auction applied on the highest $\ell$ bidders. So the winner set is always a subset of the highest $\ell$ bidders.
\end{proof}

For $\EFOl\cdot\mathbb{I}[\#(\bv\neq\bhv)\leq 1]$ benchmark, we use the $OPT_{\X}(\bhv)$ mechanism.
\begin{mechanism}\label{mech:lsa 1}
    Run the $OPT(\bhv)$ mechanism. That is, offer price $\hat{v}_{\sigma(j)}$ for $j\leq \ell$ and reject bidder $\sigma(j)$ for $j>\ell$.
\end{mechanism}

\begin{lemma}\label{lem:lsa 1}
    \mech{lsa 1} is $1$-consistent against $OPT_{\X}$ and $2$-robust against $\EFOl\cdot\mathbb{I}[\#(\bv\neq\bhv)\leq 1]$.
\end{lemma}
\begin{proof}
    \textbf{Truthfulness}, $\textbf{Feasibility}$ and $\textbf{Consistency}$ are obvious.

      \textbf{Robustness:} If there is exactly one wrong prediction, denote this bidder as $j$. Consider bidder $i$ satisfying $i\neq j$ and $i\leq \ell$, the expected revenue from $i$ is $v_i$, so the total revenue is at least
        \begin{align*}
            \sum_{i\leq \ell, i\neq j} v_i&\geq \sum_{i=2}^{\ell} v_i \geq \frac{1}{2}(v_2+ \sum_{i=2}^{\ell} v_i) \geq \frac{1}{2} \EFOl(\boldsymbol{v}),
        \end{align*}
        which shows the $2$-robustness ratio against $\EFOl\cdot\mathbb{I}[\#(\bv\neq\bhv)\leq 1]$.
\end{proof}

Combining equation \eq{lsa decomposition}, \lem{employee}, \lem{lsa 2}, \lem{lsa 1} and \lem{benchmark decomposition}, the following theorem is immediate.
\begin{theorem}\label{thm:limited supply efo}
    There is a mechanism with $1$-consistency ratio against $OPT_{\X}$ and $5.42$-robustness ratio against $\EFOl$. 
\end{theorem}

\section{Downward-Closed Permutation Environments}\label{sect:downward-closed}
In this section, we propose the mechanism for the general downward-closed permutation environment auction. In this general environment, the literature only considers $\EFOt{\X}$ benchmark and we take it as our benchmark for the robustness ratio. 

Without loss of generality, we assume that for any bidder, the allocation vector $\boldsymbol{e}_i$ that only serves bidder $i$ with probability $1$ is feasible. To begin with, we first give a decomposition of the benchmark. Note that $EFO_{\X}$ is sub-additive and monotone, so
\begin{align*}
    \EFOt{\X}(\bv)&=EFO_{\X}(\boldsymbol{v}^{(2)})\leq EFO_{\X}((v_2,v_2,\underbrace{0,\ldots,0}_{n-2}))+ EFO_{\X}^{(3)}(\bv)\leq 2v_2 +EFO_{\X}^{(3)}(\bv).
\end{align*}
Then, we have,
\begin{align}
    EFO^{(2)}_{\mathcal{X}}(\boldsymbol{v}) &\leq  \underbrace{EFO^{(2)}_{\mathcal{X}}(\boldsymbol{v})\cdot \mathbb{I}[\#(\boldsymbol{v}\neq \hat{\boldsymbol{v}})\geq 3]}_{\mbox{Multiple Errors Benchmark}}+\underbrace{2v_2\cdot \mathbb{I}[\#(\boldsymbol{v}\neq \hat{\boldsymbol{v}})=2]}_{\mbox{Sensitive Benchmark}}\notag
    \\&\quad\quad +\underbrace{EFO^{(2)}_{\mathcal{X}}(\boldsymbol{v})\cdot \mathbb{I}[\#(\boldsymbol{v}\neq \hat{\boldsymbol{v}})\leq 1] +EFO^{(3)}_{\mathcal{X}}(\boldsymbol{v})\cdot \mathbb{I}[\#(\boldsymbol{v}\neq \hat{\boldsymbol{v}})=2}_{\mbox{Insensitive Benchmarks}}].
\end{align}
As you can see, we decompose the benchmark into $3$ parts, and we will design learning-augmented mechanisms for these $3$ parts respectively. 

Before diving into these mechanisms, we first introduce an important mechanism that will replace the mechanism $OPT_{\X}(\bhv)$ for ensuring the consistency ratio in our mechanism for downward-closed environments.
The mechanism is a modified variant of \emph{weakest competitor VCG mechanism}~\cite{krishna1998efficient}, we call it \emph{Discard-and-Limit Weakest Competitor VCG Mechanism}. Compared with $OPT_{\X}(\bhv)$, the allocation rule of this mechanism is computed according to both the bids and predictions rather than only predictions. This fact is beneficial to the design of our bid-independent combination mechanism\footnote{Using $OPT_{\X}(\bhv)$ to guarantee consistency is also possible, but the resulting mechanism and proof are more sophisticated.}.

\paragraph{Discard-and-Limit Weakest Competitor VCG Mechanism}\footnote{The naming of this mechanism may seem confusing at first, but the reader should understand the naming after we define its error-tolerant version in Section \ref{sect:error-tolerant}.} Given the predicted value vector $\bhv$, we first define a Discard-and-Limit operator $\mathtt{DaL}_{\bhv}$. Given the input bid vector $\boldsymbol{b}$, $\mathtt{DaL}_{\bhv}$ outputs a bid vector $\widetilde{\boldsymbol{b}}$ such that $\widetilde{b}_i=0$ if $b_i< \hat{v}_i$; $\widetilde{b}_i=\hat{v}_i$ if $b_i\geq \hat{v}_i$. Without ambiguity, we also use $\mathtt{DaL}_{\bhv}(\boldsymbol{b})$ to represent $\boldsymbol{\tilde{b}}$. Let $\boldsymbol{x}_{\mathtt{DWC}\mbox{-}\mathtt{VCG}},\boldsymbol{p}_{\mathtt{DWC}\mbox{-}\mathtt{VCG}}$ denote the allocation and the payment of the mechanism respectively. The mechanism first changes the bid vector using a Discard-and-Limit operator $\mathtt{DaL}_{\bhv}$. Then the mechanism uses the allocation that maximizes the social welfare with bid vector $\mathtt{DaL}_{\bhv}(\bb)$, that is, $\boldsymbol{x}_{\mathtt{DWC}\mbox{-}\mathtt{VCG}}=\argmax_{\boldsymbol{x}\in \X}\sum_{i\in [n]} x_i\cdot \mathtt{DaL}_{\hat{v}_i}(b_i)$. Note that we always break ties such that $x_{\mathtt{DWC}\mbox{-}\mathtt{VCG},i}=0$ if $\mathtt{DaL}_{\hat{v}_i}(b_i)=0$. The payment $p_{\mathtt{DWC}\mbox{-}\mathtt{VCG},i}$ of the bidder $i$ is 
\begin{equation}\label{eq:payment no ET}
    p_{\mathtt{DWC}\mbox{-}\mathtt{VCG},i} = \max_{\boldsymbol{x}\in\X} \left(x_i\cdot \hat{v}_i +\sum_{j\neq i} x_j\cdot \mathtt{DaL}_{\hat{v}_j}(b_j)\right)-\sum_{j\neq i}x_{\mathtt{DWC}\mbox{-}\mathtt{VCG},j}\cdot \mathtt{DaL}_{\hat{v}_j}(b_j).
\end{equation}

To ensure the individual rationality, if $i\in \mathcal{I}:= \{i\mid x_{\mathtt{DWC}\mbox{-}\mathtt{VCG},i}\cdot b_i < p_{\mathtt{DWC}\mbox{-}\mathtt{VCG},i} \}$, we set both $x_{\mathtt{DWC}\mbox{-}\mathtt{VCG},i}$ and $p_{\mathtt{DWC}\mbox{-}\mathtt{VCG},i}$ to $0$. At last, we use $\mathtt{DWC\mbox{-}VCG}_{\bhv}$ to denote the Discard-and-Limit Weakest Competitor VCG Mechanism with parameter $\bhv$.

We first show some properties of the Discard-and-Limit weakest competitor VCG mechanism.

\begin{lemma}\label{lem:payment lower no ET}
    Given bid vector $\boldsymbol{b}$ and $\bhv$. The payment $p_{\mathtt{DWC}\mbox{-}\mathtt{VCG},i}$ calculated from \eq{payment no ET} satisfies $p_{\mathtt{DWC}\mbox{-}\mathtt{VCG},i} \geq x_{\mathtt{DWC\mbox{-}VCG},i}\cdot \hat{v}_i$. If $b_i\geq \hat{v}_i$, then $x_{\mathtt{DWC\mbox{-}VCG},i}\cdot b_i \geq p_{\mathtt{DWC}\mbox{-}\mathtt{VCG},i}$, i.e. $i\notin \mathcal{I}$.
\end{lemma}
\begin{proof}
\begin{align*}
    p_{\mathtt{DWC}\mbox{-}\mathtt{VCG},i} &= \max_{x\in\X} \left(x_i\cdot \hat{v}_i +\sum_{j\neq i} x_j\cdot \mathtt{DaL}_{\hat{v}_j}(b_j)\right)-\sum_{j\neq i}x_{\mathtt{DWC}\mbox{-}\mathtt{VCG},j}\cdot \mathtt{DaL}_{\hat{v}_j}(b_j)
    \\&\geq x_{\mathtt{DWC}\mbox{-}\mathtt{VCG},i}\cdot \hat{v}_i +\sum_{j\neq i} x_{\mathtt{DWC}\mbox{-}\mathtt{VCG},j}\cdot \mathtt{DaL}_{\hat{v}_j}(b_j) - \sum_{j\neq i}x_{\mathtt{DWC}\mbox{-}\mathtt{VCG},j}\cdot \mathtt{DaL}_{\hat{v}_j}(b_j)
    \\&\geq x_{\mathtt{DWC\mbox{-}VCG},i}\cdot \hat{v}_i.
\end{align*}
    For the second inequality, if $b_i\geq \hat{v}_i$,
    \begin{align*}
        p_{\mathtt{DWC}\mbox{-}\mathtt{VCG},i} &= \max_{x\in\X} \left(x_i\cdot \hat{v}_i +\sum_{j\neq i} x_j\cdot \mathtt{DaL}_{\hat{v}_j}(b_j)\right)-\sum_{j\neq i}x_{\mathtt{DWC}\mbox{-}\mathtt{VCG},j}\cdot \mathtt{DaL}_{\hat{v}_j}(b_j)
    \\&= \max_{x\in\X} \left(x_i\cdot \mathtt{DaL}_{\hat{v}_j}(b_i) +\sum_{j\neq i} x_j\cdot \mathtt{DaL}_{\hat{v}_j}(b_j)\right)-\sum_{j\neq i}x_{\mathtt{DWC}\mbox{-}\mathtt{VCG},j}\cdot \mathtt{DaL}_{\hat{v}_j}(b_j)
    \\&= \sum_{j\in [n]}x_{\mathtt{DWC\mbox{-}VCG},j}\cdot \mathtt{DaL}_{\hat{v}_j}(b_j  ) -\sum_{j\neq i}x_{\mathtt{DWC\mbox{-}VCG},j}\cdot \mathtt{DaL}_{\hat{v}_j}(b_j) = x_{\mathtt{DWC\mbox{-}VCG},i}\cdot \mathtt{DaL}_{\hat{v}_i}(b_i)
    \\&\leq x_{\mathtt{DWC\mbox{-}VCG},i}\cdot b_i.
    \end{align*}
\end{proof}

\begin{lemma}\label{lem:DWC truthful no ET}
    The Discard-and-Limit Weakest Competitor VCG Mechanism is truthful.
\end{lemma}
\begin{proof}
    Fix any $b_j'$, let $\boldsymbol{x}'$ be the allocation when the bidder $j$ report $b_j'$ and other bidders report $\boldsymbol{b}_{-j}$. Therefore, $\boldsymbol{x}' = \argmax_{\boldsymbol{x}^*\in \X} \left(x_j^*\cdot \mathtt{DaL}_{\hat{v}_j}(b_j')+\sum_{i\neq j} x_i^* \cdot \mathtt{DaL}_{\hat{v}_i}(b_i)\right)$. The utility of bidder $j$ is $ x_j' \cdot v_j + \sum_{i\neq j} x_i' \cdot \mathtt{DaL}_{\hat{v}_i}(b_i) - \max_{\boldsymbol{x}^*\in \X}\left(x^*_j\cdot \hat{v}_j + \sum_{i\neq j}x^*_i\cdot \mathtt{DaL}_{\hat{v}_i}(b_i)\right) $. Since the $\max$ term is independent with $j$'s reported value, the bidder $j$ should maximize the term $x_j' \cdot v_j + \sum_{i\neq j} x_i' \cdot \mathtt{DaL}_{\hat{v}_i}(b_i)$ to maximize his utility. 
    
    If $v_j=\hat{v}_j$, then $v_j=\mathtt{DaL}_{\hat{v}_j}(v_j)$, reporting truthfully leads to the allocation which maximizes $x_j' \cdot \mathtt{DaL}_{\hat{v}_j}(v_j) + \sum_{i\neq j} x_i' \cdot \mathtt{DaL}_{\hat{v}_i}(b_i) = x_j' \cdot v_j + \sum_{i\neq j} x_i' \cdot \mathtt{DaL}_{\hat{v}_i}(b_i)$. Therefore, the bidder $j$ will report truthfully.

    If $v_j> \hat{v}_j$, we rewrite the term as
    \[x_j' \cdot (v_j-\hat{v}_j)+ \underbrace{x_j'\cdot \mathtt{DaL}_{\hat{v}_j}(v_j)+\sum_{i\neq j} x_i' \cdot \mathtt{DaL}_{\hat{v}_i}(b_i)}_{(A)}.\]
    The $(A)$ term is maximized by reporting truthfully. Now we focus on the first term and prove that it is also maximized by reporting truthfully. Note that if $j$ report $b_j'\geq \hat{v}_j$, the resulting allocation $x_j'$ is the same since the bid is truncated to $ \hat{v}_j$. Thus misreporting $b_j'\geq \hat{v}_j$ does not change the utility. Let $\boldsymbol{x}$ be the allocation when $j$ reports truthfully. If $j$ report $b_j'< \hat{v}_j$, we prove that $x_j'\leq x_j$. Suppose that $x_j'>x_j$, we have 
    \begin{equation}\label{eq:901}
        x_j' \cdot  \mathtt{DaL}_{\hat{v}_j}(b_j')+\sum_{i\neq j}x_i' \cdot  \mathtt{DaL}_{\hat{v}_i}(b_i)=\max_{\boldsymbol{x}^*\in \X}\left(x_j^* \cdot  \mathtt{DaL}_{\hat{v}_j}(b_j')+\sum_{i\neq j}x_i^* \cdot  \mathtt{DaL}_{\hat{v}_i}(b_i)\right)\geq x_j \cdot  \mathtt{DaL}_{\hat{v}_j}(b_j')+\sum_{i\neq j}x_i \cdot  \mathtt{DaL}_{\hat{v}_i}(b_i)
    \end{equation}
    and 
    \[x_j \cdot  \mathtt{DaL}_{\hat{v}_j}(v_j)+\sum_{i\neq j}x_i \cdot  \mathtt{DaL}_{\hat{v}_i}(b_i)=\max_{\boldsymbol{x}^*\in \X}\left( x_j^* \cdot  \mathtt{DaL}_{\hat{v}_j}(v_j)+\sum_{i\neq j}x^*_i \cdot  \mathtt{DaL}_{\hat{v}_i}(b_i)\right) \geq x'_j \cdot  \mathtt{DaL}_{\hat{v}_j}(v_j)+\sum_{i\neq j}x'_i \cdot  \mathtt{DaL}_{\hat{v}_i}(b_i) .\]
    Since $v_j> b_j'$ and $x_j'>x_j$, we have 
    \begin{align*}
        \sum_{i\neq j}(x_i-x_i')\cdot \mathtt{DaL}_{\hat{v}_i}(b_i) \geq (x_j'-x_j)\cdot \mathtt{DaL}_{\hat{v}_j}(v_j)> (x_j'-x_j)\cdot \mathtt{DaL}_{\hat{v}_j}(b_j')
    \end{align*}
    Rearrange the inequality, we have
    \[x_j' \cdot  \mathtt{DaL}_{\hat{v}_j}(b_j')+\sum_{i\neq j}x_i' \cdot  \mathtt{DaL}_{\hat{v}_i}(b_i)< x_j \cdot  \mathtt{DaL}_{\hat{v}_j}(b_j')+\sum_{i\neq j}x_i \cdot  \mathtt{DaL}_{\hat{v}_i}(b_i),\]
    which is contradictory to \eq{901}. Therefore, if $j$ misreports $b_j'<\hat{v}_j$, $x_j' \cdot (v_j- \hat{v}_j)$ will not increase and the first term is maximized by reporting truthfully.

    If $v_j< \hat{v}_j$, by \lem{payment lower no ET}, the unit payment is at least $ \hat{v}_j>v_j$. Truthfully reporting leads to $0$ utility. However, if $j$ reports a higher value $b_j'\geq \hat{v}_j$, the utility will be negative. If $j$ reports value $b_j' < \hat{v}_j$, it will be rejected directly and receive $0$ utility.
\end{proof}

\begin{lemma}\label{lem:DWC no ET}
     The total expected revenue of $\mathtt{DWC\mbox{-}VCG}_{\bhv}$ is at least $OPT_{\X}(\mathtt{DaL}_{\bhv}(\bv))$.
\end{lemma}
\begin{proof}
    For bidder $i$ with $v_i\geq \hat{v}_i$, by \lem{payment lower no ET}, $p_{\mathtt{DWC}\mbox{-}\mathtt{VCG},i}\geq x_{\mathtt{DWC\mbox{-}VCG},i}\cdot \hat{v}_i= x_{\mathtt{DWC\mbox{-}VCG},i}\cdot \mathtt{DaL}_{\hat{v}_i}(v_i)$. Then the total expected revenue is at least 
    \begin{align*}
        \sum_{v_i\geq \hat{v}_i}x_{\mathtt{DWC\mbox{-}VCG},i}\cdot \mathtt{DaL}_{\hat{v}_i}(v_i) &= \sum_{i\in [n]}x_{\mathtt{DWC\mbox{-}VCG},i}\cdot \mathtt{DaL}_{\hat{v}_i}(v_i)= OPT_{\X}(\mathtt{DaL}_{\bhv}(\bv))
    \end{align*}
    where the first equality is because $\mathtt{DaL}(v_i)=0$ when $v_i<\hat{v}_i$.
\end{proof}

\subsection{Handling Multiple Errors Benchmark}

Recall the \mech{dga 2}, if we use a similar version here, the outcome may be infeasible when $\#(\bv\neq\bhv) = 1$. Let $j$ be the bidder with the wrong prediction. In this case, $j$ is applied with the $\mathtt{DWC}\mbox{-}\mathtt{VCG}_{\bhv}$ mechanism, and the others are applied with the black-box mechanism. The simultaneous operation of these two different mechanisms may result in conflicts and compromise the feasibility of the allocation. However, it is inevitable to have multiple mechanisms running simultaneously. Our idea is to use a special mechanism that can run together with any other mechanism to bridge $\mathtt{DWC}\mbox{-}\mathtt{VCG}_{\bhv}$ and $\Mcal$. In the downward-closed environment, this special mechanism is to reject all the bidders. We use $\Mcal_{\emptyset}$ to denote this rejection mechanism. Then we propose the following mechanism, which is a bid-independent combination of $\mathtt{DWC}\mbox{-}\mathtt{VCG}_{\bhv},\Mcal$ and $\Mcal_{\emptyset}$.

\begin{mechanism}\label{mech:rank2 1}
    For each bidder $i$, if $\#(\boldsymbol{v}_{-i} \neq \hat{\boldsymbol{v}}_{-i})=0$, run $\mathtt{DWC\mbox{-}VCG}_{\bhv}$ on the bidder $i$. If $\#(\boldsymbol{v}_{-i} \neq \hat{\boldsymbol{v}}_{-i})=1$, reject bidder $i$. If $\#(\boldsymbol{v}_{-i} \neq \hat{\boldsymbol{v}}_{-i})\geq 2$, run $\mathcal{M}$ on bidder $i$. 
\end{mechanism}

\begin{lemma}
    \mech{rank2 1} is truthful and feasible. Under the symmetric downward-closed environment, assuming that $\mathcal{M}$ is $\alpha$-competitive against $\EFOt{\X}$, \mech{rank2 1} is $1$-consistent against $OPT_{\X}(\bv)$ and $\alpha$-robust against the benchmark $EFO^{(2)}_{\mathcal{X}}(\boldsymbol{v})\cdot \mathbb{I}[\#(\boldsymbol{v}\neq \hat{\boldsymbol{v}})\geq 3]$.
\end{lemma}
\begin{proof}
    \textbf{Truthfulness:} Since $\#(\bv_{-i}\neq \bhv_{-i})$ is independent with $b_i$ and \mech{rank2 1} is a bid-independent combination of $\mathtt{DWC\mbox{-}VCG}_{\bhv}$, $\Mcal_{\emptyset}$ and $\Mcal$, which are truthful. Then \mech{rank2 1} is truthful.
    
    \textbf{Feasibility:} When $\#(\bv\neq \bhv)=0$, \mech{rank2 1} degenerates to $\mathtt{DWC\mbox{-}VCG}_{\bhv}$ mechanism, thus it is feasible. When $\#(\bv\neq \bhv) \geq 3$, then $\forall i\in [n]$, we have $\#(\bv_{-i}\neq \bhv_{-i})\geq 2$, the mechanism degenerates to $\Mcal$ and becomes feasible. When $\#(\bv\neq \bhv) = 1$, only $\mathtt{DWC\mbox{-}VCG}_{\bhv}$ and $\Mcal_{\emptyset}$ are active, then the allocation is feasible since $\X$ is downward-closed. When $\#(\boldsymbol{v}\neq\hat{\boldsymbol{v}}) = 2$, only $\Mcal$ and $\Mcal_{\emptyset}$ are active, the allocation is also feasible.

    \textbf{Consistency and robustness:} When $\#(\bv\neq \bhv)=0$, \mech{rank2 1} degenerates to $\mathtt{DWC\mbox{-}VCG}_{\bhv}$. By \lem{DWC no ET}, the expected revenue is at least $OPT_{\X}(\mathtt{DaL}_{\bhv}(\bv)) = OPT_{\X}(\bv)$. When $\#(\boldsymbol{v}\neq \hat{\boldsymbol{v}})\geq 3$, \mech{rank2 1} degenerates to $\mathcal{M}$, thus it is $\alpha$-competitive against $EFO_{\mathcal{X}}^{(2)}(\boldsymbol{v})$.
\end{proof}
\subsection{Handling Insensitive Benchmarks}
With \mech{rank2 1}, the remaining case is $\#(\boldsymbol{v}\neq \hat{\boldsymbol{v}})= 1$ or $2$. In \sect{digital good auction}, we are only left with the case $\#(\boldsymbol{v}\neq \hat{\boldsymbol{v}})\leq 1$. Let us recall what we did to handle this case. \lem{dga 1} shows that even though $OPT_{\X}(\bhv)$ operates solely based on the predicted values $\bhv$, completely disregarding all bids, it still has a nice $2$-robustness ratio against $\F\cdot \mathbb{I}[\#(\bv\neq \bhv)=1]$ benchmark. The main reason is, since there is at most one bidder with the wrong prediction, we lose at most one bidder's revenue and $\F(\bhv)$ is somewhat insensitive to the revenue loss of one bidder. 

The insensitivity of $\F$ can be shown also exist in the benchmark $\EFOt{\X}$. In fact, we prove a stronger result showing that $EFO_{\X}^{(m)}(\bv):=EFO_{\X}(\bv^{(m)})$ is insensitive to the revenue loss of fewer than $m$ bidders. Here $\bv^{(m)} = (\underbrace{v_m,\ldots,v_m}_{m},v_{m+1},\ldots,v_n)$. To be specific, let $S\subseteq [n]$ be a set of bidders, $\bv_{-S}$ denotes the value vector by setting the value of bidders in $S$ to $0$. Then we have
\begin{lemma}\label{lem:-1to2}
    For any $\boldsymbol{v}$, symmetric $\mathcal{X}$, set $S\subseteq [n]$ and positive integers $m$ with $m>|S|$, we have $EFO_{\mathcal{X}}(\boldsymbol{v}_{-S})\geq \frac{m-|S|}{m} EFO^{(m)}_{\mathcal{X}}(\boldsymbol{v})$.
\end{lemma}
\begin{proof}
    Let $s:=|S|$, define $\boldsymbol{v}':=(v_{s+1},\ldots,v_n,\underbrace{0,\ldots,0}_{s})$. Since $EFO_{\X}(\boldsymbol{v}')\leq EFO_{\X}(\bv_{-S})$ for any set $S$ with $s$ bidders, we only need to show 
    \[EFO_{\X}(\bv')\geq \frac{m-s}{m}EFO^{}(\bv^{(m)}).\]
    Let $R:[0,n]\rightarrow \mathbb{R}$ be the non-decreasing concave envelop of $\{(i,iv_i^{(m)})\}_{i=0}^n$, which we have introduced in \lem{virtual revenue}. Similarly, let $R':[0,n]\rightarrow \mathbb{R}$ be the non-decreasing concave envelop of $\{(i,iv'_i)\}_{i=0}^n$.

    We first show that for any $j\in [n]$, $R'(j)\geq  \frac{m-s}{m}R(j)$, then we use this fact to prove the lemma. For $j\leq m-s$, $R(j)=j\cdot v_m$, $R'(j)\geq j\cdot v_{s+j}\geq j\cdot v_m=R(j)$. For $m-s<j\leq m$, $R(j)=j\cdot v_m$ and 
    \[R'(j)\geq R'(m-s)\geq (m-s)v_m = \frac{m-s}{j} R(j)\geq \frac{m-s}{m}R(j).\]

    Next, we prove $R'(j)\geq  \frac{m-s}{m}R(j)$ for $j>m$. The formal definition of $R$ and $R'$ can be written as follows and we only focus their evaluation on integer numbers, 
    \[R(j) = \max_{1\leq l\leq j}\max_{\genfrac{}{}{0pt}{}{i,k:}{1\leq i\leq l\leq k\leq n}} \left[iv_i^{(m)} \frac{k-l}{k-i} + kv_k^{(m)} \frac{l-i}{k-i}\right],\quad  R'(j) = \max_{1\leq l\leq j}\max_{\genfrac{}{}{0pt}{}{i,k:}{1\leq i\leq l\leq k\leq n}} \left[iv'_i \frac{k-l}{k-i} + kv'_k \frac{l-i}{k-i}\right]\]
    Let $i^*,k^*,l^*$ be the parameter such that
    \[R(j) = i^*v_{i^*}^{(m)} \frac{k^*-l^*}{k^*-i^*} + k^*v_{k^*}^{(m)} \frac{l^*-i^*}{k^*-i^*}\]
    It is obvious that $l^*\geq m$ thus $k^*\geq l^*\geq m$. We further require $i^*\geq m$, this is because if $i^*<m$ and $k^*\geq m$, we have 
    \[R(j)= i^* v_m \frac{k^*-l^*}{k^*-i^*} + k^*v_{k^*}\frac{l^*-i^*}{k^*-i^*}.\] 
    Replace $i^*$ with $m$, we have
    \[m v_m \frac{k^*-l^*}{k^*-m} + k^*v_{k^*}\frac{l^*-m}{k^*-m}> i^* v_m \frac{k^*-l^*}{k^*-i^*} + k^*v_{k^*}\frac{l^*-i^*}{k^*-i^*}.\]
    Which leads to a contradiction. The LHS of above inequality is a weighted average of $mv_m$ and $k^*v_{k^*}$, and the RHS is a weighted average of $i^*v_m$ and $k^*v_{k^*}$. The inequality is because $mv_m>i^*v_m$ and $\frac{k^*-l^*}{k^*-m}>\frac{k^*-l^*}{k^*-i^*}$, $\frac{l^*-m}{k^*-m}<\frac{l^*-i^*}{k^*-i^*}$. 

    With the above argument, we can find $m\leq i^*\leq l^*\leq k^*\leq n, l^*\leq j$ such that 
    \[R(j) = i^*v_{i^*}^{(m)} \frac{k^*-l^*}{k^*-i^*} + k^*v_{k^*}^{(m)} \frac{l^*-i^*}{k^*-i^*}=i^*v_{i^*} \frac{k^*-l^*}{k^*-i^*} + k^*v_{k^*} \frac{l^*-i^*}{k^*-i^*}\]
    Let $i':=i^*-s>0$, $l':=l^*-s>0$, $k':=k^*-s>0$, we have
    \begin{align*}
        R'(j)&\geq i'v'_{i'} \frac{k'-l'}{k'-i'} + k'v'_{k'} \frac{l'-i'}{k'-i'}
        \\&= (i^*-s) v_{i^*}\frac{k^*-l^*}{k^*-i^*} + (k^*-s)v_{k^*}\frac{l^*-i^*}{k^*-i^*}
        \\&\geq \min\left\{\frac{i^*-s}{i^*}, \frac{k^*-s}{k^*}\right\}\left(i^*v_{i^*} \frac{k^*-l^*}{k^*-i^*} + k^*v_{k^*} \frac{l^*-i^*}{k^*-i^*}\right)
        \\&\geq \frac{m-s}{m} R(j)
    \end{align*}

   Since $\mathcal{X}$ is symmetric, there is a monotone allocation $\boldsymbol{x}\in \mathcal{X}$ such that $EFO_{\mathcal{X}}(\boldsymbol{v}^{(m)}) = EF^{\boldsymbol{x}}(\boldsymbol{v}^{(m)})$. We next show that $EF^{\boldsymbol{x}}(\boldsymbol{v}')\geq \frac{m-s}{m}EF^{\boldsymbol{x}}(\boldsymbol{v}^{(m)})$, which is enough to prove $EFO_{\mathcal{X}}(\boldsymbol{v}')\geq \frac{m-s}{m}EFO_{\mathcal{X}}(\boldsymbol{v}^{(m)})$.
    
    By Abel's lemma(a.k.a. summation by parts),
    \begin{align*}
        EF^{\boldsymbol{x}}(\boldsymbol{v}') &= \sum_{i=1}^n (R'(i)-R'(i-1)) \cdot x_i
        \\&= x_n R'(n)+ \sum_{i=2}^n R'(i-1) (x_{i-1}-x_{i})
        \\&\geq x_n \frac{m-s}{m} R(n) + \sum_{i=2}^n \frac{m-s}{m}R(i-1) (x_{i-1}-x_{i})
        \\&= \frac{m-s}{m}\sum_{i=1}^n (R(i)-R(i-1)) \cdot x_i
        \\&=\frac{m-s}{m}EF^{\boldsymbol{x}}(\boldsymbol{v}^{(m)})
    \end{align*}
    The second and third equality is because of Abel's lemma. The first inequality is because $R'(j)\geq \frac{m-s}{m}R(j), \forall j \in [n]$ and $x_{i-1}-x_i\geq 0$.
\end{proof}
This lemma indicates that even if we ignore arbitrary an arbitrary set of $s$ bidders for $s<m$, the resulting $EFO_{\X}$ benchmark is still left with a constant proportional part compared with $EFO^{(m)}_{\X}(\bv)$. With these lemma, we prove that $\mathtt{DWC\mbox{-}VCG}_{\bhv}$ can tackle the benchmark $EFO_{\X}^{(2)}\cdot \mathbb{I}[\#(\bv\neq \bhv)= 1]$ and $EFO_{\X}^{(3)}\cdot \mathbb{I}[\#(\bv\neq \bhv)= 2]$, since $EFO_{\X}^{(2)}$ is insensitive to one wrong prediction and $EFO_{\X}^{(3)}$ is insensitive to two wrong predictions by \lem{-1to2}.

\begin{mechanism}\label{mech:insensitive benchmark}
    Run $\mathtt{DWC\mbox{-}VCG}_{\bhv}$ on all bidders.
\end{mechanism}

\begin{lemma}
    \mech{insensitive benchmark} is truthful and feasible. Moreover, it is $1$-consistent against $OPT_{\X}(\bv)$ and $3$-competitive against the benchmark $EFO_{\X}^{(2)}(\bv)\cdot\mathbb{I}[\#(\bv\neq \bhv) = 1]+EFO_{\X}^{(3)}\mathbb{I}[\#(\bv\neq \bhv) = 2]$.
\end{lemma}
\begin{proof}
    Let $\boldsymbol{x}$ be the allocation of \mech{insensitive benchmark}.
    
    \textbf{Truthfulness} is due to \lem{DWC truthful no ET}. \textbf{Feasibility} is obvious. 

    \textbf{Consistency and robustness:} By \lem{DWC no ET}, the total expected revenue is at least $OPT_{\X}(\mathtt{DaL}_{\bhv}(\bv))$. When $\#(\bv\neq \bhv)=0$, $\mathtt{DaL}_{\bhv}(\bv) = \bv$. Therefore the expected revenue is $OPT_{\X}(\bv)$.

    When $\#(\bv\neq \bhv) \geq 1$, let $S := \{k \mid v_k\neq \hat{v}_k\}$, then $\mathtt{DaL}_{\bhv}(\bv)\geq \bv_{-S}$. For any $m\geq 1$,
    \[OPT_{\X}(\mathtt{DaL}_{\bhv}(\bv))\geq OPT_{\X}(\bv_{-S})\geq EFO_{\X}(\bv_{-S})\geq \frac{m+1-|S|}{m+1}EFO^{(m+1)}_{\X}(\bv).\]
    The last inequality is because \lem{-1to2}. Let $m=2$ and $m=3$, we have the expected revenue is at least
    \begin{align*}
        &\max\left\{\frac{1}{2} EFO_{\X}^{(2)}(\bv)\cdot\mathbb{I}[\#(\bv\neq \bhv) = 1],\frac{1}{3} EFO_{\X}^{(3)}(\bv)\cdot\mathbb{I}[\#(\bv\neq \bhv) = 2] \right\} 
        \\&\geq EFO_{\X}^{(2)}(\bv)\cdot\mathbb{I}[\#(\bv\neq \bhv) = 1]+ EFO_{\X}^{(3)}(\bv)\cdot\mathbb{I}[\#(\bv\neq \bhv) = 2].
    \end{align*}
    
\end{proof}


\subsection{Handling Sensitive Benchmark}
Now we focus on the remaining $2v_2\cdot\mathbb{I}[\#(\bv\neq \bhv)= 2]$ benchmark. To tackle this benchmark, we slightly modify the single-item Vickrey auction and use the bid-independent combination of this modified version and $OPT_{\X}(\bhv)$. Since we are now facing a concrete mechanism rather than a black box, we can further leverage the properties of the single-item Vickrey auction to address the feasibility issue. 

Before proceeding with the next proofs, let us first elaborate on the tie-breaking rule of $\mathtt{DWC\mbox{-}VCG}_{\bhv}$, which is crucial for the subsequent proof. Let $\boldsymbol{x}^*$ be the allocation of $\mathtt{DWC\mbox{-}VCG}_{\bhv}$, we require that $x^*_{i}\geq x^*_{j}$ for all $i<j$. The following lemma ensures that we can meet this requirement.

\begin{lemma}[Tie-Breaking for Linear Optimization]\label{lem:tie-breaking}
    For symmetric $\X$, there exist a solution $\boldsymbol{x}^*$ such that $\boldsymbol{x}^*=\argmax_{\boldsymbol{x}\in\mathcal{X}}\sum_{i=1}^n x_{i}\cdot v_{i}$ and for any $i<j$, we have $x^*_{i}\geq x^*_{j}$.
\end{lemma}
\begin{proof}
    Since $\X$ is symmetric, then if $v_{i} > v_{j}$, we have $x^*_{i}\geq x^*_{j}$, or we can swap $x^*_{i}$ and $x^*_{j}$ and get another feasible solution that has a strictly better objective. 
    
    If there exist $i<j$ such that for any $k$ satisfying $i\leq k\leq j$, $v_i=v_k=v_{j}$, then we can always find new optimal allocation $\hat{\boldsymbol{x}}^*$ such that $\hat{x}^*_{k}:= \frac{\sum_{i\leq h\leq j} \hat{x}_{h}} {j-i+1}$, $\forall i\leq k\leq j$, and for $k>j$ and $k<i$ we let $\hat{x}^*_{k}:= x^*_{k}$. The new allocation $\hat{\boldsymbol{x}}^*$ has the same revenue as $\hat{\boldsymbol{x}}$. After modifying all such $i,j$ tuple, we get the allocation we want. 
\end{proof}



Since $2v_2$ benchmark is sensitive to two wrong predictions, we have no hope of obtaining a bounded robustness ratio if we still use a bid-independent combination of mechanisms that do not make use of the bids. Our mechanism is based on a modified version of the single-item Vickrey auction (that is, serve the bidder $1$ with probability $1$ and offer the second highest bid as the payment of the bidder $1$). After the modification, it can run concurrently with $OPT_{\X}(\bhv)$, which addresses the feasibility issue of the bid-independent combination approach. Also, it guarantees at least half the revenue of the original single-item Vickrey auction.

\paragraph{Restricted Single-Item Vickrey Auction} Let $\boldsymbol{x}^*$ be the allocation of $\mathtt{DWC\mbox{-}VCG}_{\bhv}$ and $\overline{\boldsymbol{x}}$ be the allocation of single-item Vickrey auction (i.e. $\overline{x}_1=1, \overline{x}_i=0,\forall i\neq 1$). Given any bidder $j$, we slightly modify the single-item Vickrey auction by restricting the allocation of the only winner.  Let $\rho:[n]\rightarrow [n]$ be the decreasing order of entries of $\mathtt{DaL}_{\bhv}(\bv)$, if there are two bidders $i,j$ with $\mathtt{DaL}_{\hat{v}_i}(v_i)= \mathtt{DaL}_{\hat{v}_j}(v_j)$, we break ties such that $\rho$ is consistent with the index order. That is, if $v_i=v_j$, $\mathtt{DaL}_{\hat{v}_i}(v_i)= \mathtt{DaL}_{\hat{v}_j}(v_j)$ and $i<j$, we set $\rho(i)<\rho(j)$, and vice versa. Then the allocation vector $\boldsymbol{x}^{j,\downarrow}$ of the restricted single-item Vickrey auction is defined by $x^{j,\downarrow}_1 = \max\{x_{\rho(1)}^*,1-x_j^*\}$ and $x^{j,\downarrow}_i=0, \forall i\neq 1$. When there is no ambiguity in the context, we omit $j$ and use the notation $\boldsymbol{x}^{\downarrow}$ in place of $\boldsymbol{x}^{j,\downarrow}$. This allocation rule is monotone and we can charge the threshold bidding (i.e. $v_2$) as the unit price, which ensures the truthfulness of the mechanism. We use the notation $\mathtt{ResVic}_j$ to refer to the restricted single-item Vickrey auction. 

Our mechanism for $2v_2\cdot \mathbb{I}[\#(\bhv\neq \bv)=2]$ is a bid-independent combination of $\mathtt{DWC\mbox{-}VCG}_{\bhv}$, $\mathtt{ResVic}_j$, and original single-item Vickrey auction. Therefore it is truthful.

\begin{mechanism}\label{mech:rank2 4}
    For each bidder $i$, run $\mathtt{DWC\mbox{-}VCG}_{\bhv}$ on $i$ if $\boldsymbol{v}_{-i}= \hat{\boldsymbol{v}}_{-i}$. If $\#(\boldsymbol{v}_{-i}\neq \hat{\boldsymbol{v}}_{-i})=1$, let $j\neq i$ be the bidder with the wrong prediction, run $\mathtt{ResVic}_j$ on $i$. If  $\#(\boldsymbol{v}_{-i}\neq \hat{\boldsymbol{v}}_{-i})\geq 2$, run single-item Vickrey auction on $i$.
\end{mechanism} 
\begin{lemma}\label{lem:rank2 4}
    \mech{rank2 4} outputs a feasible allocation. Moreover, \mech{rank2 4} is $1$-consistent against $OPT_{\X}$, $4$-robust against the benchmark $2v_2\cdot \mathbb{I}[\#(\bv\neq \bhv)=2]$.
\end{lemma}

\begin{proof}
    \textbf{Feasibility:} When $\#(\boldsymbol{v}\neq \hat{\boldsymbol{v}})= 1$, let $j$ be the bidder with the wrong prediction. In this case, $j$ is applied with $\mathtt{DWC\mbox{-}VCG}_{\bhv}$, and others are applied with $\mathtt{ResVic}_j$. Let the allocation be $\boldsymbol{x}$, We have $x_j = x^*_j$, $x_i = x_i^{\downarrow}, \forall i\neq j$. Note that $x_i^{\downarrow} >0$ only when $i=1$. If $j=1$, then $\boldsymbol{x}\leq \boldsymbol{e}_1$ thus $\boldsymbol{x}$ is feasible. 
    If $v_j<\hat{v}_j$, then $x_j = 0$ and the feasibility holds.
    Therefore, we assume $v_j> \hat{v}_j$ and $j\neq 1$ below.
    Under this assumption, $j\neq \rho(1)$, otherwise $v_j>\hat{v}_j=\hat{v}_{\rho(1)}=\mathtt{DaL}_{\hat{v}_{\rho(1)}}(v_{\rho(1)})\geq\mathtt{DaL}_{\hat{v}_{k}}(v_{k})=v_k, \forall k\neq j$, thus $j=1$, which is contradict to our assumption.
    By the definition of $\mathtt{ResVic}_j$, $x_1=\max\{x^*_{\rho(1)},1-x^*_{j}\}$. 
    If $x_1=x^*_{\rho(1)}$, we have $\boldsymbol{x}= x_{\rho(1)}^*\boldsymbol{e}_{1}+x_j^* \boldsymbol{e}_{j}$. Since $\rho(1)\neq j$, we have $x_{\rho(1)}^*\boldsymbol{e}_{\rho(1)}+x_j^* \boldsymbol{e}_{j}\leq \boldsymbol{x}^*\in\X$, then $\boldsymbol{x}\in \X$ by the symmetry of $\X$ since $\boldsymbol{e}_j\X$ and $\boldsymbol{e}_1\in\X$.
    If $x_1 = 1- x^*_j$, then $\boldsymbol{x}= x_j^* \boldsymbol{e}_j+ (1-x_j^*)\boldsymbol{e}_1\in \X$ by the convexity of $\X$. Therefore $\boldsymbol{x}\in \X$.
    
    When $\#(\boldsymbol{v}\neq \hat{\boldsymbol{v}})\geq 2$, all bidders are applied with the single-item Vickrey auction or its restricted version,  only $x_1>0$, so the allocation is feasible.
    
    \textbf{Consistency:} If $\boldsymbol{v} = \hat{\boldsymbol{v}}$, the mechanism degenerates to $\mathtt{DWC\mbox{-}VCG}_{\bhv}$. Thus the consistency ratio is $1$.
    
    \textbf{Robustness:}It is enough to consider the revenue when $\#(\boldsymbol{v}\neq \hat{\boldsymbol{v}})=2$. In this situation, if $\#(\boldsymbol{v}_{-1}\neq \hat{\boldsymbol{v}}_{-1})\geq 2$, then bidder $1$ is applied with single-item Vickrey auction and pay $v_2$, thus the revenue is at least $v_2$. If $\#(\boldsymbol{v}_{-1}\neq \hat{\boldsymbol{v}}_{-1})= 1$, let $j\neq 1$ be the bidder with the wrong prediction, the bidder $1$ is applied with $\mathtt{ResVic}_j$, and pays $\max\{x^*_{\rho(1)},1-x^*_j\}\cdot v_2$. Since  $\mathtt{DaL}_{\hat{v}_j}(v_j)\leq \mathtt{DaL}_{\hat{v}_{\rho(1)}}(v_{\rho(1)})$
    and \lem{tie-breaking}, we have $x^*_{\rho(1)}\geq x^*_j$, then $\max\{x^*_{\rho(1)},1-x^*_j\}\geq \max\{x_j^*, 1-x_j^*\}\geq \frac{1}{2}$, thus the revenue is at least $\frac{v_2}{2}$. Overall, the robustness ratio is at most $4$.
\end{proof}

\subsection{Overall Mechanism}
Our overall mechanism for $\EFOt{\X}$ benchmark is running \mech{rank2 1}, \mech{insensitive benchmark}, \mech{rank2 4} with probability $\frac{\alpha}{\alpha+7}$, $\frac{3}{\alpha+7}$, $\frac{4}{ 
 \alpha+7}$ respectively where $\alpha$ is the competitive ratio of the black-box mechanism. The following theorem holds immediately from \lem{benchmark decomposition}.

\begin{theorem}\label{thm:general}
    If there is an $\alpha$-competitive truthful mechanism against $\EFOt{\X}$ for the downward-closed permutation environment, then there is a truthful mechanism with $1$-consistency ratio against $OPT_{\X}$, $(\alpha+7)$-robustness ratio against $\EFOt{\X}$ for the auctions with downward-closed permutation environment.
\end{theorem}

Since there is a $6.5$-competitive truthful mechanism against $\EFOt{\X}$ for the downward-closed permutation environment \cite{chen2015competitive}, we have the following corollary.

\begin{corollary}\label{cor:general}
    There is a truthful mechanism with $1$-consistency ratio against $OPT_{\X}$, $13.5$-robustness ratio against $\EFOt{\X}$ for the auctions with downward-closed permutation environment.
\end{corollary}

\section{Online Auction}\label{sect:online auctions}
In this section, we investigate the online auction setting. Suppose that bidders arrive online in random order. Upon the arrival of each bidder $i$, the prediction $\hat{v}_i$ is revealed to the auctioneer. Then, the auctioneer posts an irrevocable price for the services based on the bids and predictions of previously arrived bidders. There is no feasibility constraint, which is the same as the digital good auctions.

Koutsoupias et al.~\cite{koutsoupias2013competitive} show a black-box reduction result that any truthful offline mechanism gives rise to an online auction in their paper. Given an offline mechanism $\Mcal$, the reduction is conducted as follows: at each arriving bidder $i$, we collect the bids $\{b_1,b_2,\ldots,b_{i-1}\}$ of previously arrived bidders. Then we run a $i$-bidders offline auction $\Mcal$ on all arrived $i$ bidders. $\Mcal$ is actually giving a price to bidder $i$ according to $\{b_1,b_2,\ldots,b_{i-1}\}$, we then provide this price to bidder $i$. This reduction is called \textit{Online Sampling Auction}. Koutsoupias and Pierrakos prove that the online sampling auction transforms an 
offline mechanism to an online one with at most twice the competitive ratio. 
\begin{lemma}[\cite{koutsoupias2013competitive}]\label{lem:online sampling auction}
    If $\Mcal$ is 
    $\alpha$-competitive against $\F$ or $\maxV$ in an offline digital good auction environment. Then the online sampling auction is $2\alpha$-competitive against $\F$ or $\maxV$ in online auctions.
\end{lemma}

In our problem, we can construct a similar reduction which is named \textit{Online Sampling Auction with Predictions}.

\paragraph{Online Sampling Auction with Predictions(OSAP)} Recall that $\pi$ is the random arriving order. Upon the arrival of a bidder $\pi(t)$ and its prediction $\hat{v}_{\pi(t)}$. We run an offline mechanism $\Mcal$ on all the arrived bidders and match each prediction $\hat{v}_{\pi(i)}$ with the bidder $\pi(i)$ for $i\leq t$. Specially, we always offer the predicted price $\hat{v}_{\pi(1)}$ to the first bidder.

We first show the consistency and truthfulness of OSAP.
\begin{lemma}\label{lem:online auction}
    If the employed $\Mcal$ is truthful and $1$-consistent against $OPT$. Then the OSAP is truthful and $1$-consistent against $OPT$.
\end{lemma}
\begin{proof}
    Truthful mechanism $\Mcal$ is equivalent to proposing a random price according to the data $\{(b_{\pi(i)},\hat{v}_{\pi(i)})\}_{i<t}$ and $\hat{v}_{\pi_t}$, thus the OSAP runs like proposing an immediate price to each arrival bidder. This shows the truthfulness of the OSAP.

    Let $\bv^{[t]}_{\pi} = (v_{\pi(1)},v_{\pi(2)},\ldots,v_{\pi(t)})$ and $\bhv^{[t]}_{\pi}=(\hat{v}_{\pi(1)},\hat{v}_{\pi(2)},\ldots,\hat{v}_{\pi(t)})$, and $\Mcal^{[t]}_{\pi}$ the auction running on the first $t$ bidders and predictions. If $\bv=\bhv$, then for any realized random order $\pi$ and any arrival time $t$, we have $\bv^{[t]}_{\pi} = \bhv^{[t]}_{\pi}$. Consider the $t$-th arrival, the total revenue of the $\Mcal^{[t]}_{\pi}$ is $OPT(\bv^{[t]}_{\pi})$ by the perfect consistency of $\Mcal$. Thus the bidder $\pi(t)$ must pay his private value $v_{\pi(t)}$. Therefore, even for any realized $\pi$, the total revenue of the OSAP is $\sum_{t=1}^n v_{\pi(t)}=OPT(\bv)$. Which shows the perfect consistency of the OSAP.
\end{proof}

Since $\Mcal$ is $\alpha$-robust, i.e. $\alpha$-competitive, we can use \lem{online sampling auction} to obtain the $2\alpha$-robustness ratio of OSAP. Then \lem{online sampling auction} leads to the following corollary.


\begin{corollary}\label{cor:online auction}
    If the employed $\Mcal$ is $\alpha$-robust against $\F$ or $\maxV$, then OSAP is $2\alpha$-robust against $\F$ or $\maxV$.
\end{corollary}
Recall \cor{digital good auction}, combined with \lem{online auction}, \cor{online auction}, we have
\begin{theorem}\label{thm:online auction}
    There is a mechanism for online auctions with $1$-consistent against $OPT$ and $8.84$-robust against $\F$. Moreover, there is a mechanism for online auctions with $1$-consistent against $OPT$ and $2(e+1)$-robust against $\maxV$.
\end{theorem}
 
\begin{remark}
    In this section, our reduction is from the offline $1$-consistent auction to the online $1$-consistent auction, which is different from the reductions in previous sections. This makes our technique in this section independent of the one introduced in \sect{techniques} and more similar to the existing reduction techniques. The reason for this difference is that online auctions have a key distinction from offline auctions. In an online auction, since we can only see the arrived bids, even if the number of wrong predictions is large, we still cannot guarantee that all bidders are applied with the competitive black-box mechanism. More specifically, before the first bidder with the wrong prediction has arrived, we can only offer the predicted price to maintain perfect consistency.
\end{remark}

\section{Error-Tolerant Design}\label{sect:error-tolerant}
Though the mechanisms proposed in the previous sections have a perfect $1$-consistency guarantee, they fail if the prediction is imperfect. We cannot guarantee a good competitive ratio with $OPT$ benchmark even if the error is very small. In other words, our mechanisms have very poor error-tolerant ability. Our task in this section is to give a more practical error-tolerant version of the mechanisms we propose. Our error-tolerant mechanisms are obtained by modifying the mechanism in the previous sections.

There are two ideas used in our modification. The first is, given a confidence level $\gamma\geq 1$, we replace $\#(\bv_{-i}\neq \bhv_{-i})$ with $\#_{\gamma}(\bv_{-i}\neq \bhv_{-i})$ in all mechanisms. Here $\#_{\gamma}(\bv_{-i}\neq \bhv_{-i}):= \#\left\{j\neq i \mid  v_{j} > \gamma\hat{v}_{j} \mbox{ or } v_{j} < \frac{1}{\gamma}\hat{v}_{j} \right\}$. The second is that we replace the sub-mechanism used to guarantee perfect consistency with another mechanism that guarantees a good competitive ratio when $v_i\in [\frac{1}{\gamma} \hat{v}_i, \gamma \hat{v}_i], \forall i\in [n]$. In this section, we present the error-tolerant mechanisms for digital good auctions and auctions with downward-closed permutation environments respectively.

Before we present our mechanisms, motivated by the random expansion technique proposed by Balcan et al.~\cite{balcan2023bicriteria}, we first define a distribution of price. Given parameters $\lambda>1$, $\bhv$, $\gamma$, we let distribution $\mathcal{P}_{\bhv,\lambda,\gamma}$ be a uniform distribution over $\left\{\frac{\lambda^i}{\gamma}\hat{v}\mid 
\frac{\lambda^i}{\gamma}\hat{v}\leq \gamma \hat{v}, i\in \mathbb{Z}\cap [0,+\infty) \right\}$. Let $p$ be sampled from $\mathcal{P}_{\bhv,\lambda,\gamma}$, then $\Pr\left(p = \frac{\lambda^i}{\gamma}\hat{v}\right) = \frac{1}{\lceil\log_{\lambda}(\gamma^2)\rceil}$ when $\hat{v}>0$.

Let $\mathcal{M}$ be an $\alpha$-competitive mechanism for digital good auctions, we modify \mech{dga 2} and \mech{dga 1} as follows.
\begin{framed}
    \begin{etmechanism}{mech:dga 2}\label{mech:et dga 2}
    We first run the black-box mechanism $\mathcal{M}$ and calculate the allocation rule and payment rule of every bidder. For each bidder $i$, if $\#_{\gamma}(\boldsymbol{v}_{-i} \neq \hat{\boldsymbol{v}}_{-i})=0$, then sample a price $p_i$ from $\mathcal{P}_{\hat{v}_i,\lambda,\gamma}$ and offer it to bidder $i$; if $\#_{\gamma}(\boldsymbol{v}_{-i} \neq \hat{\boldsymbol{v}}_{-i})\geq 1$, we apply the black-box mechanism $\mathcal{M}$ on bidder $i$.
    \end{etmechanism}

    \begin{etmechanism}{mech:dga 1}\label{mech:et dga 1}
    Sample price $p_i$ from $\mathcal{P}_{\hat{v}_i,\lambda,\gamma}$ and offer $p_i$ to bidder $i$.
    \end{etmechanism}
\end{framed}

\begin{lemma}\label{lem:et dga 2}
    If $\Mcal$ is $\alpha$-competitive against $f$, then \mech{et dga 2} is truthful, and it is $(\min\{\gamma^2,\lambda\}\lceil\log_{\lambda}(\gamma^2)\rceil)$-competitive against $OPT(\bv)\cdot \mathbb{I}[\#_{\gamma}(\boldsymbol{v}\neq \hat{\boldsymbol{v}})= 0]$ and $\alpha$-competitive against $f(\bv)\cdot \mathbb{I}[\#_{\gamma}(\boldsymbol{v}\neq \hat{\boldsymbol{v}})\geq 2]$.
\end{lemma}
\begin{proof}
    Since $\#_{\gamma}(\bv_{-i}\neq \bhv_{-i})$ is independent with $b_i$ for any $\gamma$, the mechanism is still a bid-independent combination. All the combined mechanisms are truthful, including the mechanism applied when $ \#_{\gamma}(\boldsymbol{v}_{-i} \neq \hat{\boldsymbol{v}}_{-i})=0$, which is a bid-independent pricing. Therefore, the truthfulness is straightforward.

    When $\#_{\gamma}(\bv\neq \bhv)=0$, each bidder $i$ is offered with price $p_i$ that is sampled from $\mathcal{P}_{\hat{v}_i,\lambda,\gamma}$. Since $v_i\in [\frac{1}{\gamma}\hat{v}_i,\gamma \hat{v}_i]$, there exist $k\in \mathbb{Z}$ such that $v_i\in[\frac{\lambda^k}{\gamma}\hat{v}_i,\frac{\lambda^{k+1}}{\gamma}\hat{v}_i]$ and $\frac{\lambda^k}{\gamma}\hat{v}_i \in \mbox{supp}(\mathcal{P}_{\hat{v}_i,\lambda,\gamma})$. By the definition of $\mathcal{P}_{\hat{v}_i,\lambda,\gamma}$, $\Pr(p_i = \frac{\lambda^k}{\gamma}\hat{v}_i )=\frac{1}{\lceil\log_{\lambda}(\gamma^2)\rceil}$. Conditioning on that $p_i = \frac{\lambda^k}{\gamma}\hat{v}_i$ happens, the bidder $i$ pays $\frac{\lambda^k}{\gamma}\hat{v}_i\geq \frac{v_i}{\min\{\gamma^2,\lambda\}}$. Therefore, the expected revenue of the bidder $i$ is at least $\frac{v_i}{\min\{\gamma^2,\lambda\} \lceil\log_{\lambda}(\gamma^2)\rceil}$, the expected total revenue is at least $\frac{\sum_{i\in [n]} v_i}{\min\{\gamma^2,\lambda\} \lceil\log_{\lambda}(\gamma^2)\rceil}= \frac{OPT(\bv)}{\min\{\gamma^2,\lambda\}\lceil\log_{\lambda}(\gamma^2)\rceil}$.

    When $\#_{\gamma}(\bv\neq \bhv)\geq 2$, the mechanism degenerates to $\mathcal{M}$, thus it guarantees a revenue of at least $f(\bv)/\alpha$.
\end{proof}

\begin{lemma}
    \mech{et dga 1} is truthful and it is $(\min\{\gamma^2,\lambda\}\lceil\log_{\lambda}(\gamma^2)\rceil)$-competitive against $OPT(\bv)\cdot \mathbb{I}[\#_{\gamma}(\boldsymbol{v}\neq \hat{\boldsymbol{v}})= 0]$ and $2(\min\{\gamma^2,\lambda\}\lceil\log_{\lambda}(\gamma^2)\rceil)$-competitive against $f(\bv)\cdot \mathbb{I}[\#_{\gamma}(\boldsymbol{v}\neq \hat{\boldsymbol{v}})\leq 1]$.
\end{lemma}
\begin{proof}
    The truthfulness and the competitive ratio when $\#_{\gamma}(\boldsymbol{v}\neq \hat{\boldsymbol{v}})= 0$ have been shown in \lem{et dga 2}. When $\#_{\gamma}(\boldsymbol{v}\neq \hat{\boldsymbol{v}})= 1$, let $j$ denote the bidder with $v_j\notin [\frac{1}{\gamma}\hat{v}_j, \gamma \hat{v}_j]$. Since for each bidder $i\neq j$, we have $v_i\in [\frac{1}{\gamma}\hat{v}_i, \gamma \hat{v}_i]$. Then by the same analysis, the bidder $i$ will pay at least $\frac{v_i}{\min\{\gamma^2,\lambda\}\lceil\log_{\lambda}(\gamma^2)\rceil}$ in expectation. Thus the expected revenue is at least
    \[\frac{\sum_{i\neq j}v_i}{\min\{\gamma^2,\lambda\}\lceil\log_{\lambda}(\gamma^2)\rceil}\geq \frac{\sum_{i\geq 2}v_i}{\min\{\gamma^2,\lambda\}\lceil\log_{\lambda}(\gamma^2)\rceil}\geq \frac{OPT^{(2)}(\bv)}{2\min\{\gamma^2,\lambda\}\lceil\log_{\lambda}(\gamma^2)\rceil}\geq \frac{f(\bv)}{2\min\{\gamma^2,\lambda\}\lceil\log_{\lambda}(\gamma^2)\rceil}.\]
    Thus we get the competitive ratio when $ \#_{\gamma}(\boldsymbol{v}\neq \hat{\boldsymbol{v}})\leq 1$.
\end{proof}

Run \mech{et dga 2} with probability $\frac{\alpha}{\alpha+2\min\{\gamma^2,\lambda\}\lceil\log_{\lambda}(\gamma^2)\rceil}$ and \mech{et dga 1} with probability $\frac{2\min\{\gamma^2,\lambda\}\lceil\log_{\lambda}(\gamma^2)\rceil}{\alpha+ 2\min\{\gamma^2,\lambda\}\lceil\log_{\lambda}(\gamma^2)\rceil}$, since we have the decomposition \[f(\bv) = f(\bv)\cdot \mathbb{I}[\#_{\gamma}(\boldsymbol{v}\neq \hat{\boldsymbol{v}})\leq 1]+ f(\bv)\cdot \mathbb{I}[\#_{\gamma}(\boldsymbol{v}\neq \hat{\boldsymbol{v}})\geq 2],\] we obtain the following theorem immediately.
\begin{theorem}[error-tolerant design of digital good auctions]\label{thm:et dga}
    Define the predicted error rate 
  \[\eta = \max_{i\in [n]} \left\{\frac{v_i}{\hat{v}_i},\frac{\hat{v}_i}{v_i} \right\}.\] If benchmark $f$ is dominated by $OPT^{(2)}$ and there is an $\alpha$-competitive truthful mechanism against $f$, then given parameter $\lambda$ and confidence level $\gamma$, there is a truthful mechanism $\mathcal{M}^{\bhv}_{\lambda,\gamma}$ for digital good auctions, whose revenue $\mathcal{M}^{\bhv}_{\lambda,\gamma}(\bv)$ satisfies:
    \begin{itemize}
        \item If $\eta \leq \gamma$, $\mathcal{M}^{\bhv}_{\lambda,\gamma}(\bv)\geq \frac{OPT(\bv)}{\min\{\gamma^2,\lambda\}\lceil\log_{\lambda}(\gamma^2)\rceil}$.
        \item If $\eta >\gamma$, $\mathcal{M}^{\bhv}_{\lambda,\gamma}(\bv)\geq \frac{f(\bv)}{\alpha+2\min\{\gamma^2,\lambda\}\lceil\log_{\lambda}(\gamma^2)\rceil}$.
    \end{itemize}
\end{theorem}

Next, we present the error-tolerant mechanism for the auctions with downward-closed environments. We first define the error-tolerant version of our Discard-and-Limit Weakest competitor VCG mechanism.

\paragraph{Discard-and-Limit Weakest Competitor VCG Mechanism (Error-Tolerant)} Given confidence level $\gamma$, the mechanism accepts a predicted value $\hat{v}_i$ and a private value lower bound $v_i^{\downarrow}\in [\frac{1}{\gamma}\hat{v}_i,\gamma \hat{v}_i]$ as the parameters for each bidder $i$. The Discard-and-Limit operator is redefined here, and we use $\mathtt{DaL}_{\bhv,\gamma}$ to denote the new  Discard-and-Limit operator as it accepts a new parameter $\gamma$ compared to the one in \sect{downward-closed}. Given the input bid vector $\boldsymbol{b}$, $\mathtt{DaL}_{\bhv,\gamma}$ outputs a bid vector $\widetilde{\boldsymbol{b}}$ such that $\widetilde{b}_i=0$ if $b_i< \frac{1}{\gamma}\hat{v}_i$; $\widetilde{b}_i=\min\{b_i,\gamma \hat{v}_i\}$ if $b_i\geq \frac{1}{\gamma}\hat{v}_i$. Without ambiguity, we also use $ \mathtt{DaL}_{\bhv,\gamma}(b_i)$ to represent $\tilde{b}_i$. The mechanism preprocesses the bid vector using  $\mathtt{DaL}_{\bhv,\gamma}$. Denote $\bv^{\downarrow}: =(v_1^{\downarrow},v_2^{\downarrow},\ldots,v_n^{\downarrow})$. Let $\boldsymbol{x}_{\mathtt{DWC}\mbox{-}\mathtt{VCG}},\boldsymbol{p}_{\mathtt{DWC}\mbox{-}\mathtt{VCG}}$ denote the allocation and the payment of the mechanism respectively. The mechanism uses the allocation that maximizes the social welfare with bid vector $\mathtt{DaL}_{\bhv,\gamma}(\bb)$, that is, $\boldsymbol{x}_{\mathtt{DWC}\mbox{-}\mathtt{VCG}}=\argmax_{\boldsymbol{x}\in \X}\sum_{i\in [n]} x_i\cdot \mathtt{DaL}_{\bhv,\gamma}(b_i)$, here we use $\mathtt{DaL}_{\bhv,\gamma}(b_i)$ to denote the $i$-th bid of the bid vector $\mathtt{DaL}_{\bhv,\gamma}(\boldsymbol{b})$. Note that we always break ties such that $x_{\mathtt{DWC}\mbox{-}\mathtt{VCG},i}=0$ if $\mathtt{DaL}(b_i)=0$. The payment $p_{\mathtt{DWC}\mbox{-}\mathtt{VCG},i}$ of the bidder $i$ is 
\begin{equation}\label{eq:payment}
    p_{\mathtt{DWC}\mbox{-}\mathtt{VCG},i} = \max_{\boldsymbol{x}\in\X} \left(x_i\cdot v_i^{\downarrow} +\sum_{j\neq i} x_j\cdot \mathtt{DaL}_{\bhv,\gamma}(b_i)\right)-\sum_{j\neq i}x_{\mathtt{DWC}\mbox{-}\mathtt{VCG},j}\cdot \mathtt{DaL}_{\bhv,\gamma}(b_i).
\end{equation}

To ensure the individual rationality, if $i\in \mathcal{I}:= \{i\mid x_{\mathtt{DWC}\mbox{-}\mathtt{VCG},j}\cdot b_i < p_{\mathtt{DWC}\mbox{-}\mathtt{VCG},i} \}$, we set both $x_{\mathtt{DWC}\mbox{-}\mathtt{VCG},j}$ and $p_{\mathtt{DWC}\mbox{-}\mathtt{VCG},j}$ to $0$. At last, we use $\mathtt{DWC\mbox{-}VCG}_{\bv^{\downarrow},\gamma,\bhv}$ to denote the error-tolerant Discard-and-Limit Weakest Competitor VCG Mechanism with parameter $\bv^{\downarrow}$, $\gamma$ and $\bhv$.

Let $\mathcal{M}$ be an $\alpha$-competitive mechanism for auctions with downward-closed permutation environments, we modify \mech{rank2 1}, \mech{insensitive benchmark} and \mech{rank2 4} as follows.

\begin{framed}
    \paragraph{Pre-sampling} Before all the following mechanisms start, we first sample $v^{\downarrow}_i$ from $\mathcal{P}_{\hat{v}_i,\lambda,\gamma}$ independently for each $i$. This pre-sampling procedure determines the parameter $\bv^{\downarrow}$ used for the Discard-and-Limit Weakest Competitor VCG Mechanism. 

    \begin{etmechanism}{mech:rank2 1}\label{mech:et rank2 1}
    For each bidder $i$, if $\#_{\gamma}(\boldsymbol{v}_{-i} \neq \hat{\boldsymbol{v}}_{-i})=0$, run $\mathtt{DWC\mbox{-}VCG}_{\bv^{\downarrow},\gamma,\bhv}$ mechanism on the bidder $i$. If $\#_{\gamma}(\boldsymbol{v}_{-i} \neq \hat{\boldsymbol{v}}_{-i})=1$, reject bidder $i$. If $\#_{\gamma}(\boldsymbol{v}_{-i} \neq \hat{\boldsymbol{v}}_{-i})\geq 2$, run $\mathcal{M}$ on bidder $i$. 
    \end{etmechanism}

    \begin{etmechanism}{mech:insensitive benchmark}\label{mech:et insensitive benchmark}
    Run $\mathtt{DWC\mbox{-}VCG}_{\bv^{\downarrow},\gamma,\bhv}$ on all bidders.
    \end{etmechanism}

    \begin{etmechanism}{mech:rank2 4}\label{mech:et rank2 4}
    For each bidder $i$, run $\mathtt{DWC\mbox{-}VCG}_{\bv^{\downarrow},\gamma,\bhv}$ on $i$ if $\#_{\gamma}(\boldsymbol{v}_{-i}\neq \hat{\boldsymbol{v}}_{-i})=0$. If $\#_{\gamma}(\boldsymbol{v}_{-i}\neq \hat{\boldsymbol{v}}_{-i})=1$, let $j$ be the bidder with $v_j\notin [\frac{1}{\gamma}v_j,\gamma v_j]$, run $\mathtt{ResVic}_j$ on $i$ and the order $\rho$ used in $\mathtt{ResVic}_j$ should be modified to the decreasing order of entries in $\mathtt{DaL}_{\bhv,\gamma}(\bv)$. If $\#_{\gamma}(\boldsymbol{v}_{-i}\neq \hat{\boldsymbol{v}}_{-i})\geq 2$, run single-item Vickrey auction on $i$.
    \end{etmechanism} 
\end{framed}

Our error-tolerant mechanism is to run \mech{et rank2 1}, \mech{et insensitive benchmark}, \mech{et rank2 4} with probability $\frac{\alpha}{\alpha+4+3\min\{\gamma^2,\lambda\}\lceil\log_{\lambda}(\gamma^2)\rceil}$, $\frac{4}{\alpha+4+3\min\{\gamma^2,\lambda\}\lceil\log_{\lambda}(\gamma^2)\rceil}$ and $\frac{3\min\{\gamma^2,\lambda\}\lceil\log_{\lambda}(\gamma^2)\rceil}{\alpha+4+3\min\{\gamma^2,\lambda\}\lceil\log_{\lambda}(\gamma^2)\rceil}$ respectively.

\begin{theorem}[error-tolerant design of auctions with downward-closed environments]\label{thm:et dc}
    Define the predicted error rate $\eta$ as in \thm{et dga}. If there is an $\alpha$-competitive truthful mechanism against $\EFOt{\X}$, then given parameter $\lambda$ and confidence level $\gamma$, there is a truthful mechanism $\mathcal{M}^{\bhv}_{\lambda,\gamma}$ for auctions with downward-closed permutation environments, whose revenue $\mathcal{M}^{\bhv}_{\lambda,\gamma}(\bv)$ satisfies:
    \begin{itemize}
        \item If $\eta \leq \gamma$, $\mathcal{M}^{\bhv}_{\lambda,\gamma}(\bv)\geq \frac{OPT_{\X}(\bv)}{\min\{\gamma^2,\lambda\}\lceil\log_{\lambda}(\gamma^2)\rceil}$.
        \item If $\eta >\gamma$, $\mathcal{M}^{\bhv}_{\lambda,\gamma}(\bv)\geq \frac{\EFOt{\X}(\bv)}{\alpha+4+3\min\{\gamma^2,\lambda\}\lceil\log_{\lambda}(\gamma^2)\rceil}$.
    \end{itemize}
\end{theorem}
The proof of \thm{et dc} can be found in \append{A}.

\section{Lower Bound for Perfect Consistency Mechanism}\label{sect:lower bound}
We prove the lower bound for the digital good auction in the $2$-bidders case and we do not order the private value vector in this section. Let $\mathcal{M}$ be a truthful mechanism with $1$-consistency ratio. Fix the prediction $\hat{\boldsymbol{v}} = (1,N)$, here $N$ is a number to be determined later. We prove our lower bound by a probabilistic method argument. 
We consider the distribution of the private value vector of bidders. The distribution is defined as a conditional distribution of the so-called equal-revenue distribution, which is used to prove the competitive ratio lower bound of the digital good auction. The probability density function of the equal-revenue distribution is $g(\boldsymbol{v}) = \frac{1}{v_1^2 v_2^2}, \forall \boldsymbol{v}\in [1,\infty)^2$, we denote the distribution $\mathcal{G}$. Define $V := [\sqrt{N}, N]^2$, the distribution we need is simply the conditional distribution of the equal-revenue distribution given that the event $\boldsymbol{v}\in V$ happens. We denote the distribution $\mathcal{G}_{|_V}$ and its probability density $g_{|_V}$. Formally, since 
\[\Pr_{\boldsymbol{v}\sim \mathcal{G}}(\boldsymbol{v}\in V) = \int_{\sqrt{N}}^N \int_{\sqrt{N}}^N \frac{1}{v_1^2 v_2^2}dv_1dv_2=\left(\frac{\sqrt{N}-1}{N}\right)^2,\]
we have
\[g_{|_V}(\boldsymbol{v}) = \frac{g(\boldsymbol{v})}{\Pr_{\boldsymbol{v}\sim \mathcal{G}}(\boldsymbol{v}\in V)}=\frac{N^2}{(\sqrt{N}-1)^2}\frac{1}{v_1^2 v_2^2}, \quad \forall \boldsymbol{v}\in V.\]

Our lower bound proof analyzes the expected revenue that the mechanism obtains from each bidder, let $\mathcal{M}_{i}^{\bhv}(\boldsymbol{v})$ denote the expected revenue that $\mathcal{M}^{\bhv}$ obtains from the bidder $i$ when the private value vector is $\bv$. 
There are some facts about $\mathcal{M}^{\bhv}_i(\boldsymbol{v})$ shown in previous works \cite{goldberg2004lower,chen2014optimal} which is useful in our proofs, and we restate them in the following lemma. 

\begin{lemma}[\cite{goldberg2004lower,chen2014optimal}]\label{lem:revenue lemma}
    Let $\mathcal{M}^{\bhv}$ be a truthful mechanism, then
    \begin{itemize}
        \item[(i)] $\mathcal{M}^{\bhv}_i(v_i;\boldsymbol{v}_{-i})$ is non-decreasing with respect to $v_i$.
        \item[(ii)] $\mathcal{M}^{\bhv}_i(\boldsymbol{v}) \leq v_i$.
        \item[(iii)] $\int_1^{\infty} \frac{1}{v_i^2}\mathcal{M}^{\bhv}_i(v_i;\boldsymbol{v}_{-i})d v_i \leq 1$.
    \end{itemize}
\end{lemma}
We first bound the expected revenue of $\mathcal{M}^{\bhv}$ and expectation of $\F$ benchmark on instance distribution $\mathcal{G}_{|_V}$ in \lem{benchmark lower bound} and \lem{revenue upper bound}. Then we use a probabilistic method argument to show the lower bound in \thm{lower bound F}.

A direct calculation gives the following lemma,
\begin{lemma}\label{lem:benchmark lower bound}
    \[\mathbb{E}_{\boldsymbol{v}\sim \mathcal{G}_{|_V}}\left[\mathcal{F}^{(2)}(\boldsymbol{v})\right]  = \frac{4N}{\sqrt{N}-1} - \frac{2N \ln N}{(\sqrt{N}-1)^2}.\]
\end{lemma}
\begin{proof}
  Since $\mathcal{F}^{(2)}(\boldsymbol{v}) = 2\min\{v_1,v_2\}$, we have
    \begin{align*}
      \mathbb{E}_{\boldsymbol{v}\sim \mathcal{G}_{|_V}}\left[\mathcal{F}^{(2)}(\boldsymbol{v})\right] &= \int_{\sqrt{N}}^{N} \int_{\sqrt{N}}^{N} 2 g_{|_V}(v_1,v_2) \min\{v_1,v_2\} dv_1dv_2
      \\&= 2\int_{\sqrt{N}}^N dv_2\int_{v_2}^{N} \frac{N^2}{(\sqrt{N}-1)^2}\frac{2v_2}{v_1^2 v_2^2} dv_1
      \\& =\frac{2N^2}{(\sqrt{N}-1)^2}\int_{\sqrt{N}}^N\left(\frac{1}{v_2}-\frac{1}{N}\right)\frac{2}{v_2} dv_2
      \\& = \frac{4N}{\sqrt{N}-1} - \frac{4N \ln(\sqrt{N})}{(\sqrt{N}-1)^2}=\frac{4N}{\sqrt{N}-1} - \frac{2N \ln N}{(\sqrt{N}-1)^2}.
    \end{align*}
\end{proof}

Before we bound $\mathbb{E}_{\boldsymbol{v}\sim \mathcal{G}_{|_V}}\left[\mathcal{M}^{\bhv}(\boldsymbol{v}) \right]$ via the similar techniques developed by Goldberg et al.~\cite{goldberg2004lower}, we first present the following lemma which capture an important property of the revenue of the bidder $1$ given that $\bhv=(1,N)$.

\begin{lemma}\label{lem:M1 revenue}
    Let $\mathcal{M}$ be a $1$-consistent, $\alpha$-robust mechanism. If we fix the prediction to be $\hat{\boldsymbol{v}}=(1,N)$, then 
    \[\mathcal{M}^{\bhv}_1(v_1,v_2)\geq \frac{2}{\alpha}, \forall v_1\in [1,\infty), \forall v_2<N.\]
\end{lemma}
\begin{proof}
    Since $\mathcal{M}$ is $1$-consistent, we have $\mathcal{M}^{\bhv}_1(\hat{\boldsymbol{v}})+\mathcal{M}^{\bhv}_{2}(\hat{\boldsymbol{v}})\geq N+1$. By \lem{revenue lemma}, $\mathcal{M}_1^{\bhv}(\hat{\boldsymbol{v}})\leq 1, \mathcal{M}_2^{\bhv}(\hat{\boldsymbol{v}})\leq N$, we have $\mathcal{M}_1^{\bhv}(\hat{\boldsymbol{v}}) = 1, \mathcal{M}_{2}^{\bhv}(\hat{\boldsymbol{v}}) = N$. Since any truthful mechanism in the digital good auction is bid-independent pricing, $\Mcal^{\bhv}$ must offer the deterministic price $1$ to the bidder $1$ when $b_2=N$ to make sure $\Mcal^{\bhv}_1(\bhv)=1$. Similarly, $\Mcal^{\bhv}$ must offer the deterministic price $N$ to the bidder $2$ when $b_1=1$ to make sure $\Mcal^{\bhv}_2(\bhv)=N$. Therefore, we have,
\begin{align}\label{eq:1andN}
\mathcal{M}_1^{\bhv}(v_1,N) = 1, \forall v_1\geq 1, \quad \mathcal{M}_{2}^{\bhv}(1,v_2) = N, \forall v_2\geq N.
\end{align}
By \lem{revenue lemma} (iii), we have
\begin{align*}
    \int_{1}^{N}\frac{1}{v_2^2}\mathcal{M}_2^{\bhv}(1,v_2)dv_2 &\leq 1-\int_{N}^{\infty}\frac{1}{v_2^2}\mathcal{M}_2^{\bhv}(1, v_2)dv_2 
    \\&=1- \int_{N}^{\infty}\frac{1}{v_2^2}N dv_2 
    \\&=1-N\times \frac{1}{N}= 0.
\end{align*}
The first equality is because of equation \eq{1andN}.
Thus, $\mathcal{M}_2^{\bhv}(1,v_2) = 0,\forall v_2\in [1,N)$. Let $\alpha$ be the robustness ratio of $\mathcal{M}$, then we have
\begin{align}
    \mathcal{M}_1^{\bhv}(\boldsymbol{v})+\mathcal{M}_{2}^{\bhv}(\boldsymbol{v})\geq \frac{1}{\alpha}\mathcal{F}^{(2)}(\boldsymbol{v}), \forall \boldsymbol{v}\in [1,\infty)^2.
\end{align}
Set $\boldsymbol{v} = (1,v_2)$ where $v_2\in [1, N)$, we have $\mathcal{M}_1^{\bhv}(1,v_2)\geq \frac{1}{\alpha} \mathcal{F}^{(2)}(\boldsymbol{v})=\frac{2}{\alpha}, \forall v_2<N$. By \lem{revenue lemma} (i), $\mathcal{M}_1^{\bhv}$ is monotone with respect to $v_1$, so we obtain the desired result.
\end{proof}

Specifically, this lemma shows that the bidder $1$ must contribute (waste) a large revenue outside $V$ to maintain the perfect consistency. Thus, the value of $\mathbb{E}_{\boldsymbol{v}\sim \mathcal{G}_{|_V}}\left[\mathcal{M}^{\bhv}(\boldsymbol{v}) \right]$ is smaller compared to the setting without the consistency requirement.

\begin{lemma}\label{lem:revenue upper bound}
  \[\mathbb{E}_{\boldsymbol{v}\sim \mathcal{G}_{|_V}}\left[\mathcal{M}^{\bhv}(\boldsymbol{v}) \right]\leq \frac{N}{\sqrt{N}-1}\left(2-\frac{1}{\sqrt{N}}-\frac{2}{\alpha}+\frac{2}{\alpha\sqrt{N}}\right).\]
\end{lemma}
\begin{proof}
  \begin{align*}
    \mathbb{E}_{\boldsymbol{v}\sim \mathcal{G}_{|_V}}\left[\mathcal{M}_1^{\bhv}(\boldsymbol{v})\right]
    &=\frac{N^2}{(\sqrt{N}-1)^2}\int_{\sqrt{N}}^N dv_2 \int_{\sqrt{N}}^N \frac{1}{v_1^2 v_2^2}\mathcal{M}^{\bhv}_1(v_1,v_2) dv_1 .
  \end{align*}
  By \lem{M1 revenue} and \lem{revenue lemma} (iii),
  \begin{align*}
    \int_{\sqrt{N}}^N \frac{1}{v_1^2 v_2^2}\mathcal{M}^{\bhv}_1(v_1,v_2) dv_1
    &=\int_{1}^{N} \frac{1}{v_1^2 v_2^2}\mathcal{M}^{\bhv}_1(v_1,v_2) dv_1-\int_{1}^{\sqrt{N}} \frac{1}{v_1^2 v_2^2}\mathcal{M}^{\bhv}_1(v_1,v_2) dv_1
    \\&\leq \frac{1}{v_2^2}-\int_{1}^{\sqrt{N}} \frac{2}{\alpha v_1^2 v_2^2}dv_1
    \\&=\frac{1}{v_2^2}\left(1-\frac{2}{\alpha}+\frac{2}{\alpha\sqrt{N}}\right).
  \end{align*}
  Thus,
  \begin{align*}
    \mathbb{E}_{\boldsymbol{v}\sim \mathcal{G}_{|_V}}\left[\mathcal{M}^{\bhv}_1(\boldsymbol{v})\right]
    &\leq \frac{N^2}{(\sqrt{N}-1)^2} \left(1-\frac{2}{\alpha}+\frac{2}{\alpha\sqrt{N}}\right)\int_{\sqrt{N}}^N \frac{1}{v_2^2} dv_2
    \\&\leq \frac{N}{\sqrt{N}-1} \left(1-\frac{2}{\alpha}+\frac{2}{\alpha\sqrt{N}}\right).
  \end{align*}
  Next, we bound $\mathbb{E}_{\boldsymbol{v}\sim \mathcal{G}_{|_V}}\left[\mathcal{M}^{\bhv}_2(\boldsymbol{v})\right]$.
  \begin{align*}
    \mathbb{E}_{\boldsymbol{v}\sim \mathcal{G}_{|_V}}\left[\mathcal{M}^{\bhv}_2(\boldsymbol{v})\right]
    &= \frac{N^2}{(\sqrt{N}-1)^2}\int_{\sqrt{N}}^N \int_{\sqrt{N}}^N \frac{1}{v_1^2 v_2^2} \mathcal{M}^{\bhv}_2(\boldsymbol{v}) dv_2dv_1.
  \end{align*}
  We consider the function $h(x) = \int_{\sqrt{N}}^{\frac{1}{x}} \frac{1}{v_2^2}\mathcal{M}^{\bhv}_2(\boldsymbol{v})dv_2$, the derivative of $h$ is $h'(x) = -\mathcal{M}^{\bhv}_2(v_1,\frac{1}{x})$. Since $\mathcal{M}^{\bhv}_2(v_1,v_2)$ is non-decreasing, $\mathcal{M}^{\bhv}_2(v_1,\frac{1}{x})$ is non-increasing w.r.t. $x$. Therefore $h'(x)$ is non-decreasing, which implies that $h(x)$ is convex. By Jensen's inequality,
  \begin{align*}
    h\left(\frac{1}{N}\right)
    &\leq \frac{\frac{1}{\sqrt{N}}-\frac{1}{N}}{\frac{1}{\sqrt{N}}}h(0)+\frac{\frac{1}{N}}{\frac{1}{\sqrt{N}}}h\left(\frac{1}{\sqrt{N}}\right)
    \\&=\left(1-\frac{1}{\sqrt{N}}\right)h(0)+\frac{1}{\sqrt{N}} h\left(\frac{1}{\sqrt{N}}\right)
    \\&=\left(1-\frac{1}{\sqrt{N}}\right)\int_{\sqrt{N}}^{\infty} \frac{1}{v_2^2}\mathcal{M}^{\bhv}_2(\boldsymbol{v})dv_2
    \\&\leq 1-\frac{1}{\sqrt{N}}.
  \end{align*}
  Then,
  \begin{align*}
    \mathbb{E}_{\boldsymbol{v}\sim \mathcal{G}_{|_V}}\left[\mathcal{M}^{\bhv}_2(\boldsymbol{v})\right]
    &= \frac{N^2}{(\sqrt{N}-1)^2}\int_{\sqrt{N}}^N \frac{h\left(\frac{1}{N}\right)}{v_1^2}dv_1
    \leq \sqrt{N}.
  \end{align*}
  Therefore, the total expected revenue is bounded as
  \begin{align*}
    \mathbb{E}_{\boldsymbol{v}\sim \mathcal{G}_{|_V}}\left[\mathcal{M}^{\bhv}(\boldsymbol{v})\right]&\leq \frac{N}{\sqrt{N}-1} \left(1-\frac{2}{\alpha}+\frac{2}{\alpha\sqrt{N}}\right)+ \sqrt{N}
    =\frac{N}{\sqrt{N}-1}\left(2-\frac{1}{\sqrt{N}}-\frac{2}{\alpha}+\frac{2}{\alpha\sqrt{N}}\right).
  \end{align*} 
\end{proof}

\begin{thmbis}{thm:lower bound F}
  Any $1$-consistent truthful mechanism for digital good auction has a robustness ratio against $\F$ no less than $3$.
\end{thmbis}
\begin{proof}
  By \lem{benchmark lower bound} and \lem{revenue upper bound},
  \begin{align*}
    \frac{\mathbb{E}_{\boldsymbol{v}\sim \mathcal{G}_{|_V}}\left[\mathcal{F}^{(2)}(\boldsymbol{v})\right]}{\mathbb{E}_{\boldsymbol{v}\sim \mathcal{G}_{|_V}}\left[\mathcal{M}^{\bhv}(\boldsymbol{v})\right]}
    &\geq \frac{\frac{4N}{\sqrt{N}-1} - \frac{2N \ln(N)}{(\sqrt{N}-1)^2}}{\frac{N}{\sqrt{N}-1}\left(2-\frac{1}{\sqrt{N}}-\frac{2}{\alpha}+\frac{2}{\alpha\sqrt{N}}\right)}\geq \frac{4-\frac{2 \ln N }{\sqrt{N}-1}}{2-\frac{1}{\sqrt{N}}-\frac{2}{\alpha}+\frac{2}{\alpha\sqrt{N}}}.
  \end{align*}

  Thus there must exist a $\boldsymbol{v}\in V$ such that,
  \[\frac{\mathcal{F}^{(2)}(\boldsymbol{v})}{\mathbb{E}\left[\mathcal{M}^{\bhv}(\boldsymbol{v})\right]}\geq \frac{4-\frac{2 \ln N }{\sqrt{N}-1}}{2-\frac{2}{\alpha}-\frac{1}{\sqrt{N}}(1-\frac{2}{\alpha})}.\]
  The expectation is on the randomness of the randomized mechanism. Note that $\alpha$ is the robustness ratio of $\mathcal{M}$, then
  \[\alpha\geq \frac{4-\frac{2 \ln N }{\sqrt{N}-1}}{2-\frac{2}{\alpha}-\frac{1}{\sqrt{N}}(1-\frac{2}{\alpha})}\Longrightarrow \alpha\geq \frac{3-\frac{\ln N}{\sqrt{N}-1}-\frac{1}{\sqrt{N}}}{1-\frac{1}{2\sqrt{N}}}.\]
  Let $N\rightarrow \infty$, we have $\alpha\geq 3$.
\end{proof}

The optimal competitive ratio in digital good auctions is $2.42$ \cite{chen2014optimal}. Our robustness lower bound shows that it is impossible to maintain the optimal worst-case ratio while having perfect consistency.

\begin{remark}
    Note that the digital good auction is a special case of the more general limited-supply auction and downward-closed permutation environments, and the $\F$ benchmark is the same benchmark with $\EFOt{}$ or $\Fl$ in the digital good auction. Therefore, our lower bound in this section is also the lower bound for more general limited-supply auctions and downward-closed permutation environments.
\end{remark}

\section{Discussion}
    In this paper, we design perfect consistent learning-augmented mechanisms for several competitive auctions. It is worth mentioning that we consider a maximization problem that has a trivial random combination approach to design a learning-augmented mechanism. The existing works on the learning-augmented mechanism mainly focus on deterministic mechanism. We proposed the Bid-Independent Combination trick, which significantly outperforms the random combination approach. We recognize it as a conceptual contribution that imperfect predictions can also be leveraged non-trivially in the randomized mechanism design and maximization problems.
    
    Since our mechanisms employ mechanisms without predictions, an interesting question is whether one can improve the robustness ratio (with perfect consistency) by using non-black-box mechanisms. Specifically, let us focus on the simplest digital good auction, what is the optimal robustness ratio with perfect consistency? We have already shown that the optimal ratio is in $[3,4.42]$. The immediate idea is to revisit the proof idea of the optimal $2.42$ competitive ratio \cite{chen2014optimal}. However, the proof heavily relies on the fact that the revenue function of the optimal mechanism can be monotone. This property is inherently unachievable by mechanisms with perfect consistency.

\section*{Aknowledgement}
    This work was supported by the National Natural Science Foundation of China Grants No. 62325210 and No. 62272441. The authors thank anonymous reviewers for their helpful suggestions.

\bibliographystyle{alphaUrlePrint}
\bibliography{mech}

\newcommand{\etalchar}[1]{$^{#1}$}
\begin{thebibliography}{IKMQP21}

\bibitem[ABG{\etalchar{+}}22]{agrawal2022learning}
Priyank Agrawal, Eric Balkanski, Vasilis Gkatzelis, Tingting Ou, and Xizhi Tan.
\newblock \href{http://dx.doi.org/10.1145/3490486.3538306}{Learning-augmented
  mechanism design: Leveraging predictions for facility location}.
\newblock In {\em Proceedings of the 23rd ACM Conference on Economics and
  Computation}, EC '22, page 497–528, New York, NY, USA, 2022. Association
  for Computing Machinery.

\bibitem[AGP20]{anand2020customizing}
Keerti Anand, Rong Ge, and Debmalya Panigrahi.
\newblock Customizing ml predictions for online algorithms.
\newblock In {\em Proceedings of the 37th International Conference on Machine
  Learning}, ICML'20. JMLR.org, 2020.

\bibitem[AMS14]{alaei2009random}
Saeed Alaei, Azarakhsh Malekian, and Aravind Srinivasan.
\newblock \href{http://dx.doi.org/10.1145/2517148}{On random sampling auctions
  for digital goods}.
\newblock {\em ACM Trans. Econ. Comput.}, 2(3), jul 2014.

\bibitem[BGT23]{balkanski2023strategyproof}
Eric Balkanski, Vasilis Gkatzelis, and Xizhi Tan.
\newblock Strategyproof scheduling with predictions.
\newblock In {\em 14th Innovations in Theoretical Computer Science Conference
  (ITCS 2023)}, volume 251, page~11, Cambridge, Massachusetts, USA, 2023.
  Schloss Dagstuhl-Leibniz-Zentrum f{\"u}r Informatik.

\bibitem[BGTZ23]{balkanski2023online}
Eric Balkanski, Vasilis Gkatzelis, Xizhi Tan, and Cherlin Zhu.
\newblock Online mechanism design with predictions.
\newblock {\em arXiv preprint arXiv:2310.02879}, 2023.

\bibitem[BPS24]{balcan2023bicriteria}
Maria-Florina Balcan, Siddharth Prasad, and Tuomas Sandholm.
\newblock Bicriteria multidimensional mechanism design with side information.
\newblock In {\em Proceedings of the 37th International Conference on Neural
  Information Processing Systems}, NIPS '23, Red Hook, NY, USA, 2024. Curran
  Associates Inc.

\bibitem[CGL14]{chen2014optimal}
Ning Chen, Nick Gravin, and Pinyan Lu.
\newblock Optimal competitive auctions.
\newblock In {\em Proceedings of the forty-sixth annual ACM symposium on Theory
  of computing}, pages 253--262, 2014.

\bibitem[CGL15]{chen2015competitive}
Ning Chen, Nick Gravin, and Pinyan Lu.
\newblock Competitive analysis via benchmark decomposition.
\newblock In {\em Proceedings of the Sixteenth ACM Conference on Economics and
  Computation}, pages 363--376, 2015.

\bibitem[DHH13]{devanur2013prior}
Nikhil~R Devanur, Bach~Q Ha, and Jason~D Hartline.
\newblock Prior-free auctions for budgeted agents.
\newblock In {\em Proceedings of the fourteenth ACM conference on Electronic
  commerce}, pages 287--304, 2013.

\bibitem[DHY15]{devanur2015envy}
Nikhil~R Devanur, Jason~D Hartline, and Qiqi Yan.
\newblock Envy freedom and prior-free mechanism design.
\newblock {\em Journal of Economic Theory}, 156:103--143, 2015.

\bibitem[FFHK05]{feige2005competitive}
Uriel Feige, Abraham Flaxman, Jason~D Hartline, and Robert Kleinberg.
\newblock On the competitive ratio of the random sampling auction.
\newblock In {\em Internet and Network Economics: First International Workshop,
  WINE 2005, Hong Kong, China, December 15-17, 2005. Proceedings 1}, pages
  878--886. Springer, 2005.

\bibitem[FGHK02]{fiat2002}
Amos Fiat, Andrew~V Goldberg, Jason~D Hartline, and Anna~R Karlin.
\newblock Competitive generalized auctions.
\newblock In {\em Proceedings of the thiry-fourth annual ACM symposium on
  Theory of computing}, pages 72--81, 2002.

\bibitem[GH03]{goldberg2003competitiveness}
Andrew~V Goldberg and Jason~D Hartline.
\newblock Competitiveness via consensus.
\newblock In {\em SODA}, volume~3, pages 215--222, 2003.

\bibitem[GHK{\etalchar{+}}06]{goldberg2006competitive}
Andrew~V Goldberg, Jason~D Hartline, Anna~R Karlin, Michael Saks, and Andrew
  Wright.
\newblock Competitive auctions.
\newblock {\em Games and Economic Behavior}, 55(2):242--269, 2006.

\bibitem[GHKS04]{goldberg2004lower}
Andrew~V Goldberg, Jason~D Hartline, Anna~R Karlin, and Michael Saks.
\newblock A lower bound on the competitive ratio of truthful auctions.
\newblock In {\em STACS 2004: 21st Annual Symposium on Theoretical Aspects of
  Computer Science, Montpellier, France, March 25-27, 2004. Proceedings 21},
  pages 644--655. Springer, 2004.

\bibitem[GHW01]{goldberg2001competitive}
Andrew~V Goldberg, Jason~D Hartline, and Andrew Wright.
\newblock Competitive auctions and digital goods.
\newblock In {\em Proceedings of the twelfth annual ACM-SIAM symposium on
  Discrete algorithms}, pages 735--744, 2001.

\bibitem[GKST22]{gkatzelis2022improved}
Vasilis Gkatzelis, Kostas Kollias, Alkmini Sgouritsa, and Xizhi Tan.
\newblock Improved price of anarchy via predictions.
\newblock In {\em Proceedings of the 23rd ACM Conference on Economics and
  Computation}, pages 529--557, 2022.

\bibitem[HH12]{ha2012biased}
Bach~Q Ha and Jason~D Hartline.
\newblock The biased sampling profit extraction auction.
\newblock {\em arXiv preprint arXiv:1206.4955}, 2012.

\bibitem[HH13]{ha2013mechanism}
Bach~Q Ha and Jason~D Hartline.
\newblock Mechanism design via consensus estimates, cross checking, and profit
  extraction.
\newblock {\em ACM Transactions on Economics and Computation (TEAC)},
  1(2):1--15, 2013.

\bibitem[HM05]{hartline2005optimal}
Jason~D Hartline and Robert McGrew.
\newblock From optimal limited to unlimited supply auctions.
\newblock In {\em Proceedings of the 6th ACM Conference on Electronic
  Commerce}, pages 175--182, 2005.

\bibitem[HY11]{hartline2011envy}
Jason Hartline and Qiqi Yan.
\newblock Envy, truth, and profit.
\newblock In {\em Proceedings of the 12th ACM conference on Electronic
  Commerce}, pages 243--252, 2011.

\bibitem[IB22]{istrate2022mechanism}
Gabriel Istrate and Cosmin Bonchis.
\newblock Mechanism design with predictions for obnoxious facility location.
\newblock {\em arXiv preprint arXiv:2212.09521}, 2022.

\bibitem[II10]{ichiba2010averaging}
Takayuki Ichiba and Kazuo Iwama.
\newblock Averaging techniques for competitive auctions.
\newblock In {\em 2010 Proceedings of the Seventh Workshop on Analytic
  Algorithmics and Combinatorics (ANALCO)}, pages 74--81. SIAM, 2010.

\bibitem[IKMQP21]{im2021non}
Sungjin Im, Ravi Kumar, Mahshid Montazer~Qaem, and Manish Purohit.
\newblock Non-clairvoyant scheduling with predictions.
\newblock In {\em Proceedings of the 33rd ACM Symposium on Parallelism in
  Algorithms and Architectures}, pages 285--294, 2021.

\bibitem[JPS22]{jiang2022online}
Zhihao Jiang, Debmalya Panigrahi, and Kevin Sun.
\newblock Online algorithms for weighted paging with predictions.
\newblock {\em ACM Transactions on Algorithms (TALG)}, 18(4):1--27, 2022.

\bibitem[KP98]{krishna1998efficient}
Vijay Krishna and Motty Perry.
\newblock Efficient mechanism design.
\newblock {\em Available at SSRN 64934}, 1998.

\bibitem[KP13]{koutsoupias2013competitive}
Elias Koutsoupias and George Pierrakos.
\newblock On the competitive ratio of online sampling auctions.
\newblock {\em ACM Transactions on Economics and Computation (TEAC)},
  1(2):1--10, 2013.

\bibitem[LLMV20]{lattanzi2020online}
Silvio Lattanzi, Thomas Lavastida, Benjamin Moseley, and Sergei Vassilvitskii.
\newblock Online scheduling via learned weights.
\newblock In {\em Proceedings of the Fourteenth Annual ACM-SIAM Symposium on
  Discrete Algorithms}, pages 1859--1877. SIAM, 2020.

\bibitem[LV21]{lykouris2021competitive}
Thodoris Lykouris and Sergei Vassilvitskii.
\newblock Competitive caching with machine learned advice.
\newblock {\em Journal of the ACM (JACM)}, 68(4):1--25, 2021.

\bibitem[LX21]{li2021online}
Shi Li and Jiayi Xian.
\newblock Online unrelated machine load balancing with predictions revisited.
\newblock In {\em International Conference on Machine Learning}, pages
  6523--6532. PMLR, 2021.

\bibitem[PSK18]{purohit2018improving}
Manish Purohit, Zoya Svitkina, and Ravi Kumar.
\newblock Improving online algorithms via ml predictions.
\newblock {\em Advances in Neural Information Processing Systems}, 31, 2018.

\bibitem[Roh20]{rohatgi2020near}
Dhruv Rohatgi.
\newblock Near-optimal bounds for online caching with machine learned advice.
\newblock In {\em Proceedings of the Fourteenth Annual ACM-SIAM Symposium on
  Discrete Algorithms}, pages 1834--1845. SIAM, 2020.

\bibitem[XL22]{xu2022mechanism}
Chenyang Xu and Pinyan Lu.
\newblock Mechanism design with predictions.
\newblock In {\em Proceedings of the Thirty-First International Joint
  Conference on Artificial Intelligence (IJCAI-22)}, pages 571--577, 2022.

\end{thebibliography}

\appendix
\section{Missing Proofs in Section \ref{sect:error-tolerant}}
\subsection{Proof of Theorem \ref{thm:et dc}}\label{append:A}
We first show some properties of the error-tolerant Discard-and-Limit weakest competitor VCG mechanism that are similar to \lem{payment lower no ET}$\sim$\lem{DWC no ET}. The proofs are simply generalized from the proof of \lem{payment lower no ET}$\sim$\lem{DWC no ET}. For simplicity, we omit the subscript of $\mathtt{DaL}_{\bhv,\gamma}$ in this section.

\begin{lemma}\label{lem:payment lower}
    Given bid vector $\boldsymbol{b}$ and private value lower bound $\bv^{\downarrow}$. If $v^{\downarrow}_i\in [\frac{1}{\gamma}\hat{v}_i, \gamma \hat{v}_i], \forall i\in [n]$, the payment $p_{\mathtt{DWC}\mbox{-}\mathtt{VCG},i}$ calculated from \eq{payment} satisfies $p_{\mathtt{DWC}\mbox{-}\mathtt{VCG},i} \geq x_{\mathtt{DWC\mbox{-}VCG},i}\cdot v^{\downarrow}_i$. If $b_i\geq v_i^{\downarrow}$, then $x_{\mathtt{DWC\mbox{-}VCG},i}\cdot b_i \geq p_{\mathtt{DWC}\mbox{-}\mathtt{VCG},i}$.
\end{lemma}
\begin{proof}
\begin{align*}
    p_{\mathtt{DWC}\mbox{-}\mathtt{VCG},i} &= \max_{x\in\X} \left(x_i\cdot v_i^{\downarrow} +\sum_{j\neq i} x_j\cdot \mathtt{DaL}(b_j)\right)-\sum_{j\neq i}x_{\mathtt{DWC}\mbox{-}\mathtt{VCG},j}\cdot \mathtt{DaL}(b_j)
    \\&\geq x_{\mathtt{DWC}\mbox{-}\mathtt{VCG},i}\cdot v_i^{\downarrow} +\sum_{j\neq i} x_{\mathtt{DWC}\mbox{-}\mathtt{VCG},j}\cdot \mathtt{DaL}(b_j) - \sum_{j\neq i}x_{\mathtt{DWC}\mbox{-}\mathtt{VCG},j}\cdot \mathtt{DaL}(b_j)
    \\&\geq x_{\mathtt{DWC\mbox{-}VCG},i}\cdot v^{\downarrow}_i.
\end{align*}
    For the second inequality,
    \begin{align*}
        p_{\mathtt{DWC}\mbox{-}\mathtt{VCG},i} &= \max_{x\in\X} \left(x_i\cdot v_i^{\downarrow} +\sum_{j\neq i} x_j\cdot \mathtt{DaL}(b_j)\right)-\sum_{j\neq i}x_{\mathtt{DWC}\mbox{-}\mathtt{VCG},j}\cdot \mathtt{DaL}(b_j)
    \\&\leq \max_{x\in\X} \left(x_i\cdot \mathtt{DaL}(b_i) +\sum_{j\neq i} x_j\cdot \mathtt{DaL}(b_j)\right)-\sum_{j\neq i}x_{\mathtt{DWC}\mbox{-}\mathtt{VCG},j}\cdot \mathtt{DaL}(b_j)
    \\&= \sum_{j\in [n]}x_{\mathtt{DWC\mbox{-}VCG},i}\cdot \mathtt{DaL}(b_j  ) -\sum_{j\neq i}x_{\mathtt{DWC\mbox{-}VCG},i}\cdot \mathtt{DaL}(b_j) = x_{\mathtt{DWC\mbox{-}VCG},i}\cdot \mathtt{DaL}(b_i)
    \\&\leq x_{\mathtt{DWC\mbox{-}VCG},i}\cdot b_i.
    \end{align*}
\end{proof}

\begin{lemma}\label{lem:DWC truthful}
    The Error-Tolerant Discard-and-Limit Weakest Competitor VCG Mechanism is truthful.
\end{lemma}
\begin{proof}
    Let $\boldsymbol{x}'$ be the allocation when the bidder $j$ report $b_j'$ and other bidders report $\boldsymbol{b}_{-j}$. Therefore, $\boldsymbol{x}' = \argmax_{\boldsymbol{x}^*\in \X} \left(x_j^*]\cdot \mathtt{DaL}(b_j')+\sum_{i\neq j} x_i^* \cdot \mathtt{DaL}(b_i)\right)$. The utility of bidder $j$ is $ x_j' \cdot v_j + \sum_{i\neq j} x_i' \cdot \mathtt{DaL}(b_i) - \max_{\boldsymbol{x}^*\in \X}\left(x^*_j\cdot \hat{v}_j + \sum_{i\neq j}x^*_i\cdot \mathtt{DaL}(b_i)\right) $. Since the $\max$ term is independent with $j$'s reported value, the bidder $j$ should maximize the term $x_j' \cdot v_j + \sum_{i\neq j} x_i' \cdot \mathtt{DaL}(b_i)$ to maximize his utility. 
    
    If $v_j\in [\frac{1}{\gamma}\hat{v}_j,\gamma\hat{v}_j]$, then $v_j=\mathtt{DaL}(v_j)$, reporting truthfully leads to the allocation which maximizes $x_j' \cdot \mathtt{DaL}(v_j) + \sum_{i\neq j} x_i' \cdot \mathtt{DaL}(b_i) = x_j' \cdot v_j + \sum_{i\neq j} x_i' \cdot \mathtt{DaL}(b_i)$. Therefore, the bidder $j$ will report truthfully.

    If $v_j>\gamma \hat{v}_j$, we rewrite the term as
    \[x_j' \cdot (v_j-\gamma \hat{v}_j)+ \underbrace{x_j'\cdot \mathtt{DaL}(v_j)+\sum_{i\neq j} x_i' \cdot \mathtt{DaL}(b_i)}_{(A)}.\]
    The $(A)$ term is maximized by reporting truthfully. Now we focus on the first term and prove that it is also maximized by reporting truthfully. Note that if $j$ report $b_j'\geq \gamma \hat{v}_j$, the resulting allocation $x_j'$ is the same since the bid is truncated to $\gamma \hat{v}_j$. Thus misreporting $b_j'\geq \gamma\hat{v}_j$ does not change the utility. Let $\boldsymbol{x}$ be the allocation when $j$ reports truthfully. If $j$ report $b_j'< \gamma \hat{v}_j$, we prove that $x_j'\leq x_j$. Suppose that $x_j'>x_j$, we have 
    \begin{equation}\label{eq:90}
        x_j' \cdot  \mathtt{DaL}(b_j')+\sum_{i\neq j}x_i' \cdot  \mathtt{DaL}(b_i)=\max_{\boldsymbol{x}^*\in \X}\left(x_j^* \cdot  \mathtt{DaL}(b_j')+\sum_{i\neq j}x_i^* \cdot  \mathtt{DaL}(b_i)\right)\geq x_j \cdot  \mathtt{DaL}(b_j')+\sum_{i\neq j}x_i \cdot  \mathtt{DaL}(b_i)
    \end{equation}
    and 
    \[x_j \cdot  \mathtt{DaL}(v_j)+\sum_{i\neq j}x_i \cdot  \mathtt{DaL}(b_i)=\max_{\boldsymbol{x}^*\in \X}\left( x_j^* \cdot  \mathtt{DaL}(v_j)+\sum_{i\neq j}x^*_i \cdot  \mathtt{DaL}(b_i)\right) \geq x'_j \cdot  \mathtt{DaL}(v_j)+\sum_{i\neq j}x'_i \cdot  \mathtt{DaL}(b_i) .\]
    Since $v_j> b_j'$ and $x_j'>x_j$, we have 
    \begin{align*}
        \sum_{i\neq j}(x_i-x_i')\cdot \mathtt{DaL}(b_i) \geq (x_j'-x_j)\cdot \mathtt{DaL}(v_j)> (x_j'-x_j)\cdot \mathtt{DaL}(b_j')
    \end{align*}
    Rearrange the inequality, we have
    \[x_j' \cdot  \mathtt{DaL}(b_j')+\sum_{i\neq j}x_i' \cdot  \mathtt{DaL}(b_i)< x_j \cdot  \mathtt{DaL}(b_j')+\sum_{i\neq j}x_i \cdot  \mathtt{DaL}(b_i),\]
    which is contradictory to \eq{90}. Therefore, if $j$ misreports $b_j'<\gamma \hat{v}_j$, $x_j' \cdot (v_j-\gamma \hat{v}_j)$ will not increase and the first term is maximized by reporting truthfully.

    If $v_j< \frac{1}{\gamma}\hat{v}_j$, by \lem{payment lower}, the unit payment is at least $ v_i^{\downarrow} \geq \frac{1}{\gamma}\hat{v}_j>v_j$, if $j$ report a higher value $b_j'\geq v_i^{\downarrow}$, the utility will be negative. If $j$ report value $b_j' < v_i^{\downarrow}$, it will be rejected directly and receive $0$ utility.
\end{proof}

\begin{lemma}\label{lem:DWC}
     If $v^{\downarrow}_i$'s are sampled from $\mathcal{P}_{\hat{v}_i,\lambda,\gamma}$ independently, Then the total expected revenue of $\mathtt{DWC\mbox{-}VCG}_{\bv^{\downarrow},\gamma,\bhv}$ is at least $\frac{OPT_{\X}(\mathtt{DaL}(\bv))}{\min\{\gamma^2,\lambda\}\lceil \log_{\lambda}(\gamma^2)\rceil}$.
\end{lemma}
\begin{proof}
    Let $\widetilde{v}_i^{\downarrow}$ be the maximal value in $\mbox{supp}(\mathcal{P}_{\hat{v}_i,\lambda,\gamma})$ such that $\widetilde{v}_i^{\downarrow}\leq v_i$. 
    
    For bidder $i$ with $v_i\geq \frac{1}{\gamma}\hat{v}_i$, we have $\widetilde{v}_i^{\downarrow}\geq \frac{1}{\min\{\gamma^2,\lambda\}}\max\{v_i,\gamma\hat{v}_i\}= \frac{1}{\min\{\gamma^2,\lambda\}}\mathtt{DaL}(v_i)$. By \lem{payment lower}, when $v_i^{\downarrow}=\widetilde{v}_i^{\downarrow}$, $p_{\mathtt{DWC}\mbox{-}\mathtt{VCG},i}\geq x_{\mathtt{DWC\mbox{-}VCG},i}\cdot v^{\downarrow}_i\geq \frac{1}{\min\{\gamma^2,\lambda\}}x_{\mathtt{DWC\mbox{-}VCG},i}\cdot \mathtt{DaL}(v_i)$. Since $\Pr(v_i^{\downarrow}=\widetilde{v}_i^{\downarrow})=\frac{1}{\lceil\log_{\lambda}(\gamma^2)\rceil}$, the expected payment of bidder $i$ is at least $\frac{1}{\min\{\gamma^2,\lambda\}\lceil\log_{\lambda}(\gamma^2)\rceil}x_{\mathtt{DWC\mbox{-}VCG},i}\cdot \mathtt{DaL}(v_i)$. Then the total expected revenue is at least 
    \begin{align*}
        \sum_{v_i\geq \frac{1}{\gamma}\hat{v}_i}\frac{1}{\min\{\gamma^2,\lambda\}\lceil\log_{\lambda}(\gamma^2)\rceil}x_{\mathtt{DWC\mbox{-}VCG},i}\cdot \mathtt{DaL}(v_i) &= \sum_{i\in [n]}\frac{1}{\min\{\gamma^2,\lambda\}\lceil\log_{\lambda}(\gamma^2)\rceil}x_{\mathtt{DWC\mbox{-}VCG},i}\cdot \mathtt{DaL}(v_i)
        \\&= \frac{1}{\min\{\gamma^2,\lambda\}\lceil\log_{\lambda}(\gamma^2)\rceil} OPT_{\X}(\mathtt{DaL}(\bv)).
    \end{align*}
    Where the first equality is because $\mathtt{DaL}(v_i)=0$ when $v_i<\frac{1}{\gamma}\hat{v}_i$.
\end{proof}

Similar to the proof of \thm{general}, we first prove the revenue guarantee of three sub-mechanisms.

\begin{lemma}
    \mech{et rank2 1} is truthful and feasible. Under the symmetric environment, assuming that $\mathcal{M}$ is $\alpha$-competitive against $\EFOt{\X}$, \mech{et rank2 1} is $(\min\{\gamma^2,\lambda\}\lceil \log_{\lambda}(\gamma^2) \rceil)$-competitive against $OPT_{\X}(\bv)\cdot \mathbb{I}[\#_{\gamma}(\bv\neq \bhv)=0]$ and $\alpha$-competitive against the benchmark $EFO^{(2)}_{\mathcal{X}}(\boldsymbol{v})\cdot \mathbb{I}[\#_{\gamma}(\boldsymbol{v}\neq \hat{\boldsymbol{v}})\geq 3]$.
\end{lemma}
\begin{proof}
    \textbf{Truthfulness:} Since $\#_{\gamma}(\bv_{-i}\neq \bhv_{-i})$ is independent with $b_i$ and \mech{et rank2 1} is a bid-independent combination of $\mathtt{DWC\mbox{-}VCG}_{\bv^{\downarrow},\gamma,\bhv}$, $\Mcal_{\emptyset}$ and $\Mcal$, which are truthful. Then \mech{et rank2 1} is truthful.
    
    \textbf{Feasibility:} When $\#_{\gamma}(\bv\neq \bhv)=0$, \mech{et rank2 1} degenerates to $OPT_{\X}^{\lambda,\gamma}(\bhv)$ mechanism, thus it is feasible. When $\#_{\gamma}(\bv\neq \bhv) \geq 3$, then $\forall i\in [n]$, we have $\#_{\gamma}(\bv_{-i}\neq \bhv_{-i})\geq 2$, the mechanism degenerates to $\Mcal$ and becomes feasible. When $\#_{\gamma}(\bv\neq \bhv) = 1$, only $\mathtt{DWC\mbox{-}VCG}_{\bv^{\downarrow},\gamma,\bhv}$ and $\Mcal_{\emptyset}$ are active, then the allocation is feasible since $\X$ is downward-closed. When $\#_{\gamma}(\boldsymbol{v}\neq\hat{\boldsymbol{v}}) = 2$, only $\Mcal$ and $\Mcal_{\emptyset}$ are active, the allocation is also feasible.

    \textbf{Competitive ratio:} When $\#_{\gamma}(\bv\neq \bhv)=0$, \mech{et rank2 1} degenerates to $\mathtt{DWC\mbox{-}VCG}_{\bv^{\downarrow},\gamma,\bhv}$. By \lem{DWC}, the expected revenue is at least $\frac{OPT_{\X}(\bv)}{\min\{\gamma^2,\lambda\}\lceil \log_{\lambda}(\gamma^2)\rceil}$. When $\#_{\gamma}(\boldsymbol{v}\neq \hat{\boldsymbol{v}})\geq 3$, \mech{et rank2 1} degenerates to $\mathcal{M}$, thus it is feasible and $\alpha$-competitive against $EFO_{\mathcal{X}}^{(2)}(\boldsymbol{v})$.
\end{proof}

\begin{lemma}
    \mech{et insensitive benchmark} is truthful and feasible. Moreover, it is $(\min\{\gamma^2,\lambda\}\lceil \log_{\lambda}(\gamma^2) \rceil)$-competitive against $OPT_{\X}(\bv)\cdot \mathbb{I}[\#_{\gamma}(\bv\neq \bhv)=0]$ and $3(\min\{\gamma^2,\lambda\}\lceil \log_{\lambda}(\gamma^2) \rceil)$-competitive against the benchmark $EFO_{\X}^{(2)}(\bv)\cdot\mathbb{I}[\#(\bv\neq \bhv) = 1]+EFO_{\X}^{(3)}\mathbb{I}[\#(\bv\neq \bhv) = 2]$.
\end{lemma}
\begin{proof}
    Let $\boldsymbol{x}$ be the allocation of \mech{et insensitive benchmark}.
    
    \textbf{Truthfulness} is due to \lem{DWC truthful}. \textbf{Feasibility} is obvious. 

    \textbf{Competitive ratio:} By \lem{DWC}, the total expected revenue is at least $\frac{OPT_{\X}(\mathtt{DaL}(\bv))}{\min\{\gamma^2,\lambda\}\lceil \log_{\lambda}(\gamma^2)\rceil}$. When $\#_{\gamma}(\bv\neq \bhv)=0$, $\mathtt{DaL}(\bv) = \bv$. Therefore the expected revenue is $\frac{OPT_{\X}(\bv)}{\min\{\gamma^2,\lambda\}\lceil \log_{\lambda}(\gamma^2)\rceil}$.

    When $\#_{\gamma}(\bv\neq \bhv) \geq 1$, let $S := \{k \mid v_k\notin [\frac{1}{\gamma}\hat{v}_k, \gamma \hat{v}_k]\}$, then $\mathtt{DaL}(\bv)\geq \bv_{-S}$. For any $m\geq 1$
    \[OPT_{\X}(\mathtt{DaL}(\bv))\geq OPT_{\X}(\bv_{-S})\geq EFO_{\X}(\bv_{-S})\geq \frac{m+1-|S|}{m+1}EFO^{(m+1)}_{\X}(\bv).\]
    Let $m=2$ and $m=3$, we have the expected revenue is at least
    \begin{align*}
        &\frac{1}{\min\{\gamma^2,\lambda\}\lceil \log_{\lambda}(\gamma^2)\rceil}\max\left\{\frac{1}{2} EFO_{\X}^{(2)}(\bv)\cdot\mathbb{I}[\#_{\gamma}(\bv\neq \bhv) = 1],\frac{1}{3} EFO_{\X}^{(3)}(\bv)\cdot\mathbb{I}[\#_{\gamma}(\bv\neq \bhv) = 2] \right\} 
        \\&\geq\frac{1}{3\min\{\gamma^2,\lambda\}\lceil \log_{\lambda}(\gamma^2)\rceil} \left(EFO_{\X}^{(2)}(\bv)\cdot\mathbb{I}[\#_{\gamma}(\bv\neq \bhv) = 1]+ EFO_{\X}^{(3)}(\bv)\cdot\mathbb{I}[\#_{\gamma}(\bv\neq \bhv) = 2]\right).
    \end{align*}
    
\end{proof}
\begin{lemma}\label{lem:et rank2 4}
    \mech{et rank2 4} is truthful and feasible. Moreover, \mech{et rank2 4} is $\min\{\gamma^2,\lambda\} \lceil \log_{\lambda}(\gamma^2) \rceil$-competitive against $OPT_{\X}(\bv)\cdot \mathbb{I}[\#_{\gamma}(\bv\neq \bhv)=0]$ and $4$-competitive against the benchmark $2v_2\cdot \mathbb{I}[\#_{\gamma}(\bv\neq \bhv)=2]$.
\end{lemma}

\begin{proof}
    \textbf{Feasibility:} When $\#_{\gamma}(\boldsymbol{v}\neq \hat{\boldsymbol{v}})= 1$, let $j$ be the bidder with $v_j\notin [\frac{1}{\gamma}\hat{v}_j,\gamma\hat{v}_j]$. In this case, $j$ is applied with $\mathtt{DWC\mbox{-}VCG}_{\bv^{\downarrow},\gamma,\bhv}$, and others are applied with $\mathtt{ResVic}_j$. Let the allocation be $\boldsymbol{x}$, then we have $x_j = x^*_j$, $x_i = x_i^{\downarrow}, \forall i\neq j$. Note that $x_i^{\downarrow} >0$ only when $i=1$. If $j=1$, then $\boldsymbol{x}\leq \boldsymbol{e}_1$ thus $\boldsymbol{x}$ is feasible. If $v_j<\frac{1}{\gamma}\hat{v}_j$, then $x_j = 0$ and the feasibility holds.
    Therefore, we assume $v_j>\gamma\hat{v}_j$ and $j\neq 1$ below.
    Under this assumption, $j\neq \rho(1)$, otherwise $v_j> \gamma\hat{v}_j=\gamma\hat{v}_{\rho(1)}= \mathtt{DaL}_{\hat{v}_\rho(1),\gamma}(v_{\rho(1)})\geq\mathtt{DaL}_{\hat{v}_k,\gamma}(v_k)=v_k, \forall k\neq j$, where the last equality comes from the fact that $v_k\in [\frac{1}{\gamma}\hat{v}_j, \gamma\hat{v}_j]$ and the definition of $\mathtt{Dal}_{\bhv,\gamma}$. Thus $j=1$, which contradicts our assumption. By the definition of $\mathtt{ResVic}_j$, $x_1=\max\{x^*_{\rho(1)},1-x^*_{j}\}$. 
    If $x_1=x^*_{\rho(1)}$, we have $\boldsymbol{x}= x_{\rho(1)}^*\boldsymbol{e}_{1}+x_j^*\boldsymbol{e}_{j}$.  Since $\rho(1)\neq j$, we have $x_{\rho(1)}^*\boldsymbol{e}_{\rho(1)}+x_j^* \boldsymbol{e}_{j}\leq \boldsymbol{x}^*\in\X$, then $\boldsymbol{x}\in \X$ by the symmetry of $\X$. If $x_1 = 1- x^*_j$, then $\boldsymbol{x}\leq x_j^* \boldsymbol{e}_j+ (1-x_j^*)\boldsymbol{e}_1\in \X$ by the convexity of $\X$ since $\boldsymbol{e}_j\X$ and $\boldsymbol{e}_1\in\X$. Therefore $\boldsymbol{x}\in \X$.
    
    When $\#_{\gamma}(\boldsymbol{v}\neq \hat{\boldsymbol{v}})\geq 2$, all bidders are applied with the single-item Vickrey auction or its restricted version,  only $x_1>0$, so the allocation is feasible.
    
    
    \textbf{Competitive ratio:}
    When $\#_{\gamma}(\boldsymbol{v} \neq \hat{\boldsymbol{v}})=0$, the mechanism degenerates to $\mathtt{DWC\mbox{-}VCG}_{\bv^{\downarrow},\gamma,\bhv}$. By \lem{DWC}, it is $\min\{\gamma^2,\lambda\} \lceil \log_{\lambda}(\gamma^2) \rceil$-competitive against $OPT_{\X}(\mathtt{DaL}(\bv))=OPT_{\X}(\bv)$.

    When $\#_{\gamma}(\boldsymbol{v} \neq \hat{\boldsymbol{v}})=2$, if $\#_{\gamma}(\boldsymbol{v}_{-1}\neq \hat{\boldsymbol{v}}_{-1})= 2$, then bidder $1$ is applied with single-item Vickrey auction and pay $v_2$, thus the revenue is at least $v_2$. If $\#_{\gamma}(\boldsymbol{v}_{-1}\neq \hat{\boldsymbol{v}}_{-1})= 1$, let $j\neq 1$ be the bidder with $v_j\notin [\frac{1}{\gamma}\hat{v}_j,\gamma \hat{v}_j]$, the bidder $1$ is applied with $\mathtt{ResVic}_j$, and pays $\max\{x^*_1,1-x^*_j\}\cdot v_2$. Since  $\mathtt{DaL}_{\hat{v}_j,\gamma}(v_j)\leq \mathtt{DaL}_{\hat{v}_{\rho(1)},\gamma}(v_{\rho(1)})$ and \lem{tie-breaking}, we have $x^*_{\rho(1)}\geq x^*_j$, then $\max\{x_{\rho(1)}^*, 1-x_j^*\}\geq \max\{x_j^*, 1-x_j^*\}\geq \frac{1}{2}$, thus the revenue is at least $\frac{v_2}{2}$. Overall, the competitive ratio against $2v_2\cdot \mathbb{I}[\#_{\gamma}(\bv\neq \bhv)=2]$ is at most $4$.
\end{proof}

Overall, since we have the following benchmark decomposition
\begin{align*}
    EFO^{(2)}_{\mathcal{X}}(\boldsymbol{v}) &\leq  EFO^{(2)}_{\mathcal{X}}(\boldsymbol{v})\cdot \mathbb{I}[\#_{\gamma}(\boldsymbol{v}\neq \hat{\boldsymbol{v}})\geq 3]+2v_2\cdot \mathbb{I}[\#_{\gamma}(\boldsymbol{v}\neq \hat{\boldsymbol{v}})=2]\notag
    \\&\quad\quad +EFO^{(2)}_{\mathcal{X}}(\boldsymbol{v})\cdot \mathbb{I}[\#_{\gamma}(\boldsymbol{v}\neq \hat{\boldsymbol{v}})\leq 1] +EFO^{(3)}_{\mathcal{X}}(\boldsymbol{v})\cdot \mathbb{I}[\#_{\gamma}(\boldsymbol{v}\neq \hat{\boldsymbol{v}})=2],
\end{align*}
by \lem{benchmark decomposition}, \thm{et dc} is straightforward.

\end{document}